%% file: main.tex
\documentclass[a4paper,thm-restate]{lipics-v2021}

\usepackage{amsmath,amsfonts,amssymb,amsthm,xspace}
\usepackage{graphicx}
\usepackage{paralist}
\usepackage{xcolor}

\usepackage{algorithm}
\usepackage[noend]{algpseudocode}
\usepackage{xifthen}

\renewcommand{\Pr}[2][]{\mbox{\rm Pr}_{#1}\left[#2\right]}
\newcommand{\CandidatesArg}[1]{\mathcal{Z}_{#1}}

\newcommand{\CandidateSet}[1]{Z_{#1}}

\def\curveP{P}
\def\curveQ{Q}
\def\eps{\varepsilon}
\def\RR{{\mathbb R}}
\def\NN{{\mathbb N}}
\def\ZZ{{\mathbb Z}}
\def\XX{{\mathbb X}}
\def\cD{\mathcal{D}}

\def\RSpace{{\pazocal{R}}}

\def\Dist{{\pazocal{D}}}

\def\length{\lambda}

\DeclareMathOperator{\poly}{poly}
\DeclareMathAlphabet{\pazocal}{OMS}{zplm}{m}{n}

\renewcommand{\emph}[1]{\textit{\textbf{#1}}}

\newcommand{\disk}[2][\Delta]{
    \ifthenelse{\isempty{#2}}
    {\mathrm{b}_{#1}}
    {\mathrm{b}_{#1}(#2)}
}
\newcommand{\mdisk}[3][\Delta]{
    \ifthenelse{\isempty{#3}}
    {\mathrm{b}_{#1}^{#2}}
    {\mathrm{b}_{#1}^{#2}(#3)}
}

\newcommand{\dfree}[3][\Delta]{%
    \ifthenelse{\equal{#2}{} }
    {\cD_{#1}}
    {\cD^{(#2,#3)}_{#1}}
}

\newcommand{\Param}[1]{
    {\mathbb{T}_{#1}}
}

\newcommand{\ConcreteCandidates}[1]{
    {C_{#1}}
}

\title{Faster Approximate Covering of Subcurves under the Fr\'echet Distance}
\addtocounter{linenumber}{1000}

\author{Frederik Br\"uning}{Department of Computer Science, University of Bonn, Germany}{}{}{}
\author{Jacobus Conradi}{Department of Computer Science, University of Bonn, Germany}{}{}{}
\author{Anne Driemel}{Hausdorff Center for Mathematics, University of Bonn, Germany}{}{}{}

\authorrunning{Conradi, Brüning, and Driemel}

\Copyright{Frederik Br\"uning, Jacobus Conradi, Anne Driemel} 


\acknowledgements{This work has been funded by the Deutsche Forschungsgemeinschaft (DFG, German Research Foundation) - AA 1111/2-2 (FOR 2535 Anticipating Human Behavior).}

\nolinenumbers

\begin{document}

\hideLIPIcs

\maketitle

\begin{abstract}
Subtrajectory clustering is an important variant of the trajectory clustering problem, where the start and endpoints of trajectory patterns within the collected trajectory data are not known in advance. We study this problem in the form of a set cover problem for a given polygonal curve: find the smallest number $k$ of representative curves such that any point on the input curve is contained in a subcurve that has Fréchet distance at most a given $\Delta$ to a representative curve. We focus on the case where the representative curves are line segments and approach this NP-hard problem with classical techniques from the area of geometric set cover: we use a variant of the multiplicative weights update method which was first suggested by Brönniman and Goodrich for set cover instances with small VC-dimension. 
We obtain a bicriteria-approximation algorithm that computes a set of  $O(k\log(k))$ line segments that cover a given polygonal curve of $n$ vertices under Fréchet distance at most $O(\Delta)$. We show that the algorithm runs in  $\widetilde{O}(k^2 n +  k n^3)$ time in expectation and uses $ \widetilde{O}(k n + n^3)$ space. For two dimensional input curves that are $c$-packed, we bound the expected running time  by $\widetilde{O}(k^2 c^2 n)$  and the space by $ \widetilde{O}(kn + c^2 n)$. In $\RR^d$ the dependency on $n$ instead is quadratic.
In addition, we present a variant of the algorithm that uses implicit weight updates on the candidate set and thereby achieves  near-linear running time in $n$ without any assumptions on the input curve, while keeping the same approximation bounds. This comes at the expense of a small (polylogarithmic) dependency on the relative arclength.
\end{abstract}

\newpage
\setcounter{page}{1}

\section{Introduction}
\pagenumbering{arabic}

The advancement of tracking technology  made it possible to record the movement of single entities at a large scale in various application areas ranging from vehicle navigation over sports analytics to the socio-ecological study of animal and human behaviour. The types of trajectories that are analyzed range from GPS-trajectories \cite{acmsurvey20} to full-body-motion trajectories~\cite{ionescu2013human3} and complex gestures \cite{Qiao2017RealtimeHG}, and even include the positions of the focus point of attention from a human eye \cite{ duchowski2002breadth, holmqvist2011eye}. In many such applications, a flood of data presents us with the challenging task of extracting useful information. If a long trajectory is given as a  sequence of positions in some parameter space, it is rarely known in advance which specific movement patterns occur. In particular, it is challenging to find the start and endpoints of such patterns, which is why popular clustering algorithms heuristically partition the trajectories into smaller subtrajectories. An example is the popular algorithm by Lee, Han and Whang~\cite{LeeHW07}. 
Since the criteria according to which one should detect, group and represent behaviour patterns vary greatly among different kinds of application, there are many different variants of the subtrajectory clustering problem, see also the survey papers~\cite{BuchinW20, wang2021survey, yuan2017review}. One line of research uses the well-established Fr\'echet distance to define similarity between subcurves, for example the works of Agarwal et al.~\cite{agarwal2018}, Buchin et al.~\cite{buchinGroup20} and  Akitaya et al.~\cite{akitaya2021subtrajectory}.
In an attempt to unify previous definitions of the underlying algorithmic problem, Akitaya et al.~\cite{akitaya2021subtrajectory} define the following geometric set cover problem. Given a polygonal curve, the goal is to ``cover'' the whole curve with a minimum number of simpler representative curves, such that each point of the trajectory is contained in a subcurve with small Fréchet distance to its closest representative curve. This is in line with traditional clustering formulations such as metric $k$-center, where clusters may overlap. In this paper, we study the set cover problem introduced by  Akitaya et al.\ and improve upon their results.

\subsection{Preliminaries} 
For any $n>1$, a sequence of points $p_1,\dots,p_n \in \RR^d$ defines a \emph{polygonal curve} $\curveP$ by linearly interpolating consecutive points, that is, for each $i$, we obtain the \emph{edge} $e_i:[0,1]\rightarrow\RR^d;t\mapsto (1-t) p_i + t p_{i+1}$. We may write $e_i=\overline{p_i\,p_{i+1}}$ for edges. We may think  $P$  as a continuous function  $P:[0,1]\rightarrow \RR^d$ by fixing $n$ values $0=t_1<\ldots<t_n=1$, and defining $P(t)=e_i\left(\frac{t-t_i}{t_{i+1}-t_i}\right)$ for $t_i\leq t\leq t_{i+1}$. We call the set $(t_1,\ldots,t_n)$ the \emph{vertex parameters} of the parametrized curve $P:[0,1]\rightarrow\RR^d$.
For $n=1$, we may slightly abuse notation to view a point $p_1$ in $\RR^d$ as a polygonal curve defined by an edge of length zero with $p_2=p_1$.
We call the number of vertices $n$ the \emph{complexity} of the curve. 

We define the \emph{concatenation} of two  curves $P,Q:[0,1]\rightarrow\RR^d$ with $P(1)=Q(0)$ by $P\oplus Q:[0,1]\rightarrow\RR^d$ with $(P\oplus Q)(t)=P(2t)$ if $t\leq 1/2$, and $(P\oplus Q)(t)=Q(2t-1)$ if $t\geq 1/2$. Note that the concatenation of two polygonal curves $P$ and $Q$ with vertices $p_1,\dots,p_n$ and $q_1,\dots,q_m$ such that $p_n=q_1$ is the polygonal curve defined by the vertices $p_1,\ldots,p_n,q_2,\ldots,q_m$.
For any two $a,b \in [0,1]$ we denote with $\curveP[a,b]$ the \emph{subcurve} of $\curveP$ that starts at $\curveP(a)$ and ends at $\curveP(b)$. Note, that $a>b$ is specifically allowed and results in a subcurve in reverse direction. 

We call the subcurves of edges \emph{subedges}. Let $\XX^d_{\ell} = (\RR^d)^{\ell}$, and think of
the elements of this set as the
set of all  polygonal curves of $\ell$ vertices in $\RR^d$.

For two parametrized curves $\curveP$ and $\curveQ$, we define their \emph{Fréchet distance} as
 \[ d_F(\curveP,\curveQ) = \inf_{\alpha,\beta:[0,1] \rightarrow [0,1]} \sup_{t \in [0,1]}
\| \curveP(\alpha(t)) - \curveQ(\beta(t)) \|, \] 

\noindent where $\alpha$ and $\beta$ range over all functions that are non-decreasing, surjective and continuous.  We call the pair $(\alpha,\beta)$ a traversal. Every traversal has a distance $\sup_{t \in [0,1]}\| \curveP(\alpha(t)) - \curveQ(\beta(t)) \|$ associated to it.
We call a curve $X$ in $\RR^d$ \emph{$c$-packed}, if for any point $p$ and and radius $r$, the length of $X$ inside the $2$-disk is bounded by $||X\cap \disk[r]{p}||\leq cr$, where $\disk[r]{p}=\{x\in\RR^d\mid\|p-x\|\leq r\}$.
Let $X$ be a set. We call a set $\RSpace$ where any $r \in \RSpace$ is of the form $r \subseteq X$ a \emph{set system} with \emph{ground set} $X$. 
We say a subset $A \subseteq X$ is \emph{shattered} by $\RSpace$ if for any $A' \subseteq A$ there exists an $r \in \RSpace$ such that $A'=r \cap A$. The \emph{VC-dimension} of $\RSpace$ is the maximal size of a set $A$ that is shattered by $\RSpace$.
For a given  weight function $w$ on the ground set $X$ and a real value $\eps > 0$, we say that a subset $C \subset X$ is an \emph{$\eps$-net} if every set of $\RSpace$ of weight at least $\eps \cdot w(X)$ contains at least one element of $C$. For any $A\subseteq X$, we write $w(A)$ short for $\sum_{a\in A}w(a)$.

\subparagraph{Computational Model} We describe our algorithms in the real-RAM model of computation, which allows to store real numbers and to perform simple operations in constant time on them. We call the following operations \emph{simple operations}. The arithmetic operations $+,-,\times,/$. The comparison operations $=,\neq,>,\geq,\leq,<,$ for real numbers with output $0$ or $1$. In addition to the simple operations, we allow the square-root operation. In Appendix~\ref{sec:squareroots}, we describe how to circumvent the square-root operation with little extra cost.

\subsection{Problem definition}
\label{sec:def}

We study the same problem as Akitaya, Chambers, Brüning and Driemel~\cite{akitaya2021subtrajectory}. Throughout the paper we assume that $d$ is a constant independent of $n$.

Let $\curveP:[0,1] \rightarrow \RR^d$ be a polygonal curve of $n$ vertices and let $\ell \in \mathbb{N}$ and $\Delta \in \RR$ be fixed parameters. 
Define the $\Delta$-\emph{coverage} of a set of center curves $C \subset \XX^d_{\ell}$ as follows:
\[ \Psi_{\Delta}(P,C) = \bigcup_{q \in C} ~ \bigcup_{0 \leq t \leq t' \leq 1} \{ s \in [t,t'] \mid  d_F(\curveP[t,t'], q)\leq \Delta   \}. \]

The $\Delta$-coverage corresponds to the part of the curve $P$ that is covered by the set of all subtrajectories that are within Fréchet distance $\Delta$ to some  curve in $C$.
If for some $P,C,\Delta$ it holds that
$\Psi_{\Delta}(P,C) = [0,1]$, then we call $C$ a $\Delta$-\emph{covering} of $P$. The problem we study in this paper is to find a $\Delta$-covering $C \subset \XX^d_{\ell}$ of $P$ of minimum size.
In particular, we study bicriterial approximation algorithms for this problem, which we formalize as follows.

\begin{figure}[t]
    \centering
    \includegraphics[scale=0.8]{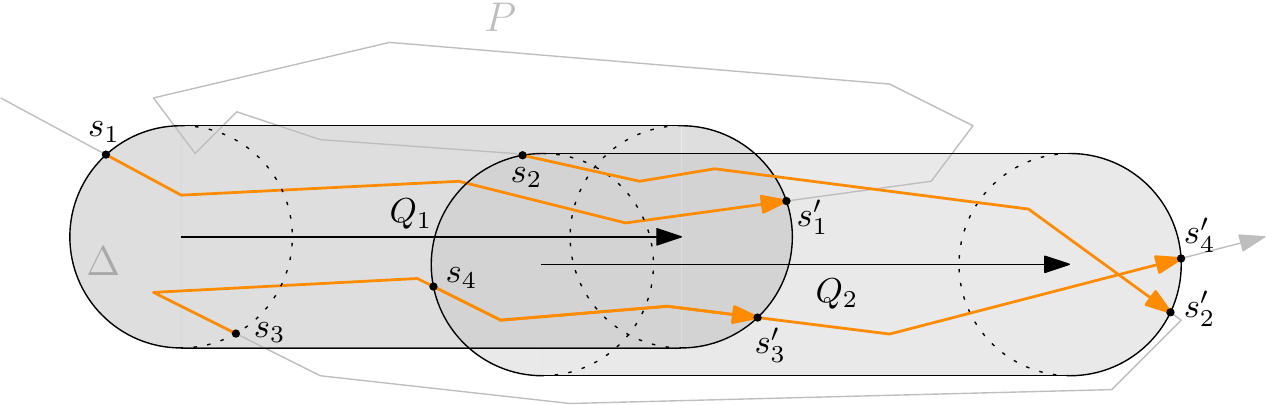}
    \caption{Illustration of the $\Delta$-coverage of a set $C=\{Q_1,Q_2\}$ and a curve $P$. Here we have $\Psi_{\Delta}(P,C)=[s_1,s'_1]\cup[s_2,s'_2]\cup[s_3,s'_4]$, since the subcurves $P[s_1,s'_1]$ and $P[s_3,s'_3]$ have Fréchet distance $\Delta$ to $Q_1$, the subcurves $P[s_2,s'_2]$ and $P[s_4,s'_4]$ have Fréchet distance $\Delta$ to $Q_2$  and each other subcurve of $P$ that has Fréchet distance at most $\Delta$ to $Q_1$ or $Q_2$ is a subcurve of $P[s_i,s'_i]$ for some $1\leq i\leq 4$.}
    \label{fig:cover}
\end{figure}

\begin{definition}[$(\alpha,\beta)$-approximate solution]\label{def:approx}
Let $P \in \XX^d_n$ be a polygonal curve, $\Delta\in\RR_+$ and $\ell \in \NN$. A set $C \subseteq \XX^d_{\ell}$ is an $(\alpha,\beta)$-approximate solution to the $\Delta$-coverage problem on $P$, if $C$ is a $\alpha\Delta$-covering of $P$ and there exists no $\Delta$-covering $C' \subseteq \XX^d_{\ell}$ of $P$ with $\beta |C'| < |C|$. 
\end{definition}

\subsection{Related work}
\input{relatedwork}

\subsection{Our contribution}
Our main result is an  algorithm that computes an $(\alpha,\beta)$-approximate solution with $\alpha \in O(1)$  and $\beta = O(\ell \log k)$, where $k$ is the size of an optimal solution.
For general curves, the algorithm runs in   $\widetilde{O}({k}^2 n +  k n^3)$ time in expectation and uses $ \widetilde{O}({k} n + n^3)$ space. (The $\widetilde{O}(\cdot)$ notation hides polylogarithmic factors in $n$ to simplify the exposition.) If the input curve is a $c$-packed polygonal curve in the plane, the expected running time can be bounded by $\widetilde{O}({k}^2 c^2  n )$  and the space is in $ \widetilde{O}({k}n + c^2 n)$.  
Our second result is an algorithm that achieves  near-linear running time in $n$---even for general polygonal curves---while keeping the same approximation bounds at the expense of a small dependency on the arclength in the running time. The algorithm needs in expectation  $\widetilde{O}(n{k}^3 \log^4(\frac{\lambda}{\Delta k}))$ time and $\widetilde{O}(n{k}\log(\frac{ \lambda}{\Delta k}))$ space, where $\lambda$ is the arclength of the input curve. Here, we stated our results for general $\ell$ using the reduction described at the end of this section. 

In our algorithms we use a variant of the multiplicative weights update method~\cite{arora2012multiplicative}, which has been used earlier for set cover problems with small VC-dimension~\cite{bronnimann1995almost}. The difficulty in our case is that the set system initially has high VC-dimension, as shown by Akitaya et al.~\cite{akitaya2021subtrajectory}---namely $\Theta(\log n)$ in the worst case. We circumvent this by defining an intermediate set cover problem where the VC-dimension is significantly reduced.  A key idea that enables our results is a curve simplification that requires the curve to be locally maximally simplified, a notion that is borrowed from de Berg, Cook, and Gudmundsson~\cite{de2013fast}. To the best of our knowledge, our candidate generation yields the first the first strongly polynomial algorithm for approximate subtrajectory clustering under the continuous Fr\'echet distance. In particular, the running time of the algorithm does not depend on the relative arclength $\lambda/\Delta$ of the input curve or the spread of the coordinates. Our second algorithm improves the dependency on the relative arclength from quadratic to polylogarithmic as compared to previous results~\cite{akitaya2021subtrajectory}. 

\subparagraph{Reduction to line segments}
In the remainder of the paper, we will focus on finding a $\Delta$-covering with line segments, that is $\ell=2$. The following lemma provides the reduction for general $\ell$ at the expense of an increased approximation factor.

\begin{lemma}\label{lem:reduction}
Let $P \in \XX^d_n$ be a polygonal curve, $\Delta\in\RR_+$ and $\ell \in \NN$. 
Let $C \subseteq \XX^d_{\ell}$ be a $\Delta$-covering of $P$ of minimum cardinality. There exists a set of line segments $C' \subseteq \XX^d_{2}$ that is a $\Delta$-covering of $P$ with $|C'| \leq (\ell-1)|C|$.
\end{lemma}

\begin{proof} Choose as set $C'$ the union of the set of edges of the polygonal curves of $C$. Clearly, this set has the claimed cardinality and is a $\Delta$-covering of $P$.
\end{proof}

\subsection{Roadmap}

In Section~\ref{sec:structure} we develop a structured variant of our problem that allows us to apply the multiplicative weight update method in the style of Brönniman and Goodrich~\cite{bronnimann1995almost} in an efficient way. Our intermediate goal in this section is to obtain a structured set of candidates for a modified coverage problem that is on the one hand easy to compute and on the other hand sufficient to obtain good approximation bounds for the original problem.
In Section~\ref{sec:simplifications}, we define a notion of curve simplification that is inspired by the work of de Berg, Gudmundsson and Cook~\cite{de2013fast}. A crucial property of this simplification is that subcurves of the input are within small Fréchet distance to subcurves of constant complexity of the simplification.
In Sections~\ref{sec:structuredCoverage}, we define a  structured notion of $\Delta$-coverage and a candidate space, which lets us take advantage of this fact. In Sections~\ref{sec:traversals} and  \ref{sec:alg_candidates}, we show we can narrow our choice down even further, to a finite set of subedges of the simplification, and still sufficiently preserve the quality of the solution.
In Section \ref{sec:sample}, we present our main algorithm. The algorithm uses the concepts and techniques developed in Section~\ref{sec:structure} in combination with the multiplicative weights update method. 
In Section \ref{sec:analysis:main} we analyze the approximation factor and running time of this algorithm. Crucially, we show that the VC-dimension of the induced set system which is implicitly used by our algorithm is small by design~(see Sections~\ref{sec:struc:feasible} to \ref{sec:feas}). We obtain results both for general as well as $c$-packed polygonal curves in Sections~\ref{sec:analysis:general} and~\ref{sec:analysis:cpacked}.
In Section \ref{sec:implicit:weights}, we present a slightly different approach to the problem. Instead of computing a finite set of candidates explicitly,  we exploit the shape of the feasible sets and show that it is possible to implicitly update the weights of a much larger set of candidates. This improves the overall dependency on the complexity of the input curve in the running time, when compared to the previous algorithm---at the cost of a logarithmic factor of the relative arc-length of the curve.
In Section~\ref{sec:conclusions}, we conclude with a brief discussion of how our techniques may be applied to other problem variants and give an overview of further research directions.

\section{Structuring the solution space}\label{sec:structure}

In this section, we introduce key concepts that allow us to transfer the problem to a set cover problem on a finite set system with small VC-dimension and still obtain good approximation bounds. The main result of this section is Theorem~\ref{thm:alg_candidates}.

\subsection{Simplifications and containers}\label{sec:simplifications}

We start by defining the notion of curve-simplification that we will use throughout the paper.
\begin{definition}[simplification]\label{def:goodsimp}
    Let $P$ be a polygonal curve in $\RR^d$. Let $(t_1,\ldots,t_n)$ be the vertex-parameters of $P$, and $p_i=P(t_i)$ the vertices of $P$. Consider an index set $1 \leq i_1<\ldots<i_k \leq n$ that defines vertices $p_{i_j}$. We call a curve $S$ defined by such an ordered set of vertices $(p_{i_1},\ldots,p_{i_k})\in(\RR^d)^k$ a \emph{simplification} of $P$. We say the simplification is $\Delta$-\emph{good}, if the following properties hold:
    \begin{compactenum}[(i)]
        \item $\|p_{i_j}-p_{i_{j+1}}\|\geq\frac{\Delta}{3}$ for $1\leq j<k$
        \item $d_F(P[t_{i_j},t_{i_{j+1}}],\overline{p_{i_j}\,p_{i_{j+1}}}) \leq 3\Delta$ for all $1\leq j<k$.
        \item $d_F(P[t_1,t_{i_1}],\overline{p_{i_1}\,p_{i_1}})\leq3\Delta$ and $d_F(P[t_{i_k},t_n],\overline{p_{i_k}\,p_{i_k}})\leq3\Delta$
        \item $ d_F(P[t_{i_j},t_{i_{j+2}}],\overline{p_{i_j}\,p_{i_{j+2}}}) > 2\Delta$ for all $1\leq j<k-1 $
    \end{compactenum}
\end{definition}
Our intuition is the following.
Property $(i)$ guarantees that $S$ does not have short edges. Property $(ii)$ and $(iii)$ together tell us, that the simplification error is small. Property $(iv)$ tells us, that the simplification is (approximately) maximally simplified, that is, we cannot remove a vertex, and hope to stay within Fr\'echet distance $2\Delta$ to $P$.

\begin{definition}[Container]\label{def:container}
    Let $P$ be a polygonal curve, let $\pi= P[s,t]$ be a subcurve of $P$, and let $(t_1,\ldots,t_n)$ be the vertex-parameters of $P$. For a simplification $S$ of $P$ defined by index set $I=(i_1,\ldots,i_k)$, define the \emph{container} $c_S(\pi)$ of $\pi$ on $S$ as $S[t_a,t_b]$, with $a=\max\left(\{i_1\}\cup\{i\in I\mid t_i\leq s\}\right)$ and $b=\min\left(\{i\in I\mid t_i\geq t\}\cup\{i_k\}\right)$.
\end{definition}

The following lemma has been proven by de Berg et al.~\cite{de2013fast}. We restate and reprove it here with respect to our notion of simplification.

\begin{restatable}{lemma}{threeedges}\cite{de2013fast}
\label{lem:three:edges}
    Let $P$ be a polygonal curve in $\RR^d$, and let $S$ be a $\Delta$-good simplification of $P$. Let $Q$ be an edge in $\RR^d$ and let $\pi$ be a subcurve of $P$ with $d_F(Q,\pi)\leq\Delta$. Then $c_S(\pi)$ consists of at most $3$ edges.
\end{restatable}

\begin{proof}
    Assume for the sake of contradiction, that $c_S(\pi)$ contains $4$ edges, that is it has three internal vertices $s_1,s_2,s_3$. By Definition~\ref{def:container} these three vertices are also interior vertices of $\pi$. As the Fréchet distance $d_F(Q,\pi)\leq\Delta$, there are points $q_1,q_2,q_3\in Q$, that get matched to $s_1,s_2$ and $s_3$ respectively during the traversal, with $\|s_i-q_i\|\leq\Delta$. This implies $d_F(\pi[s_1,s_3],\overline{q_1\,q_3})\leq\Delta$. It also implies, that $d_F(\overline{s_1\,s_3},\overline{q_1\,q_3})\leq\Delta$. But then 
    \[d_F(\overline{s_1\,s_3},P[s_1,s_3])=d_F(\overline{s_1,s_3},\pi[s_1,s_3])\leq d_F(\overline{s_1\,s_3},\overline{q_1\,q_3}) + d_F(\pi[s_1,s_3],\overline{q_1\,q_3})\leq2\Delta,\]
    
    contradicting the assumption that $S$ is a $\Delta$-good simplification.
\end{proof}

\subsection{Structured coverage and candidate space}\label{sec:structuredCoverage}\label{sec:structuredCandidates}

We want to make use of the property of $\Delta$-good simplifications shown in Lemma \ref{lem:three:edges}. For this we adapt the notion of $\Delta$-coverage from Section~\ref{sec:def} as follows.

\begin{definition}\label{def:structuredcoverage}
Let $S$ be a polygonal curve in $\RR^d$. Let $(t_1,\ldots,t_n)$ be the vertex-parameters of $S$. Let $\ell \in \mathbb{N}$ and $\Delta \in \RR$ be fixed parameters. 
Define the \emph{structured} $\Delta$-\emph{coverage} of a set of center curves $C \subset \XX^d_{\ell}$ as 
\[ \Psi'_{\Delta}(S,C) = \bigcup_{q \in C} ~ \bigcup_{(i,j)\in J}
\Psi^{(i,j)}_{\Delta}(S,q) \]
where 
\[ \Psi^{(i,j)}_{\Delta}(S,q) = 
\{ s \in [t,t'] \mid  t_i\leq t \leq t_{i+1};\, t \leq t';\, t_{j-1}\leq t'\leq t_j;\, d_F(\curveP[t,t'], q) \leq \Delta \}, \]
and where
$J = \{  1 \leq i \leq j \leq n  \mid 1 \leq j-i \leq 4 \}$. 

If it holds that $\Psi'_{\Delta}(S,C)=[0,1]$, then we call $C$ a \emph{structured} $\Delta$-\emph{covering} of $S$.
\end{definition}

\begin{figure}
    \centering
    \includegraphics[scale=0.8]{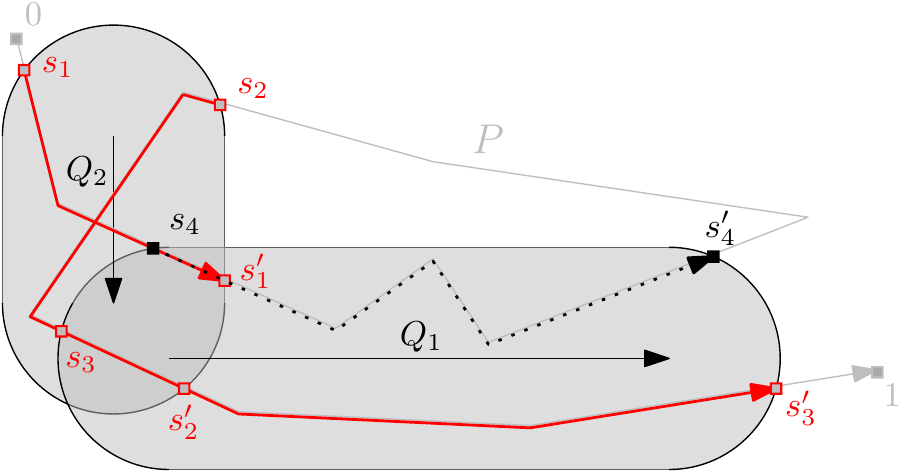}
    \caption{Example of the structured $\Delta$-coverage of a set $C=\{Q_1,Q_2\}$ and a curve $P$. Here we have $\Psi'_{\Delta}(P,C)=[s_1,s'_1]\cup[s_2,s'_3]$ since the subcurves $P[s_1,s'_1]$ and $P[s_2,s'_2]$ have Fréchet distance $\Delta$ to $Q_1$ and $P[s_3,s'_3]$ has Fréchet distance $\Delta$ to $Q_2$. Note that $[s_4,s'_4]$ is not part of the coverage since the subcurve $P[s_4,s'_4]$ consists of $4$ edges. }
    \label{fig:struct_cover}
\end{figure}

\begin{observation}\label{obs:unstructure_cover}
    In general for any polygonal curve $S$ and set of center curves $C$ it holds that $\Psi'_\Delta(S,C)\subseteq \Psi_\Delta(S,C)$. 
\end{observation}

We now want to restrict the candidate set to subedges of a simplification of the input curve, thereby imposing more structure on the solution space. For this we begin by defining a more structured parametrization of the set of edges of a polygonal curve.

\begin{definition}[Edge space]

    We define the \emph{edge space} $\Param{n}=\{1,\ldots,n-1\}\times[0,1]$.
    We denote the set of edges of $\curveP$ with $E(\curveP)$. 
\end{definition}

\begin{definition}[Candidate space]
Let $E = \{e_1,\dots,e_{n-1}\}$ be an ordered set of edges in $\RR^d$. We define the \emph{candidate space} induced by $E$ as the set  $\CandidatesArg{E} = \{ (t_1,i_1,t_2,i_2) \in \Param{n} \times \Param{n} \mid i_1=i_2\} $. We associate an element $(t_1,i,t_2,i) \in \CandidatesArg{E}$ with the subedge $\overline{e_i(t_1)\,e_i(t_2)}$. We may abuse notation by denoting the associated edge to an element $t \in \CandidatesArg{E}$ simply with $t$. 
\end{definition}

The following theorem summarizes and motivates the above definitions of structured coverage and candidate space. Namely, we can restrict the search space to subedges of the simplification $S$ and still obtain a good covering of $P$. Moreover, we can evaluate the coverage of our solution solely based on $S$. The structured coverage only allows subcurves of $S$ that consist of at most three edges to contribute to the coverage. This technical restriction is necessary to obtain a small VC-dimension in our main algorithm later on, and it is well-motivated by Lemma~\ref{lem:three:edges}.
The proof of the theorem is rather technical and we divert it to Section~\ref{sec:proofssec3}.

\begin{theorem}\label{thm:candidate_space}
Let $S$ be a $\Delta$-good simplification of a  curve $P$. Let $C$ be a set of subedges of edges of $S$. If $C$ is a structured $8\Delta$-covering of $S$, then $C$  is an $11\Delta$-covering of $P$. Moreover, if $k$ is the size of an optimal $\Delta$-covering of $P$, then there exists such a set $C$ of size at most $3k$.
\end{theorem}

\subsection{Partial traversals and coverage}\label{sec:traversals}

\begin{figure}[t]
    \centering
    \includegraphics[width=\textwidth]{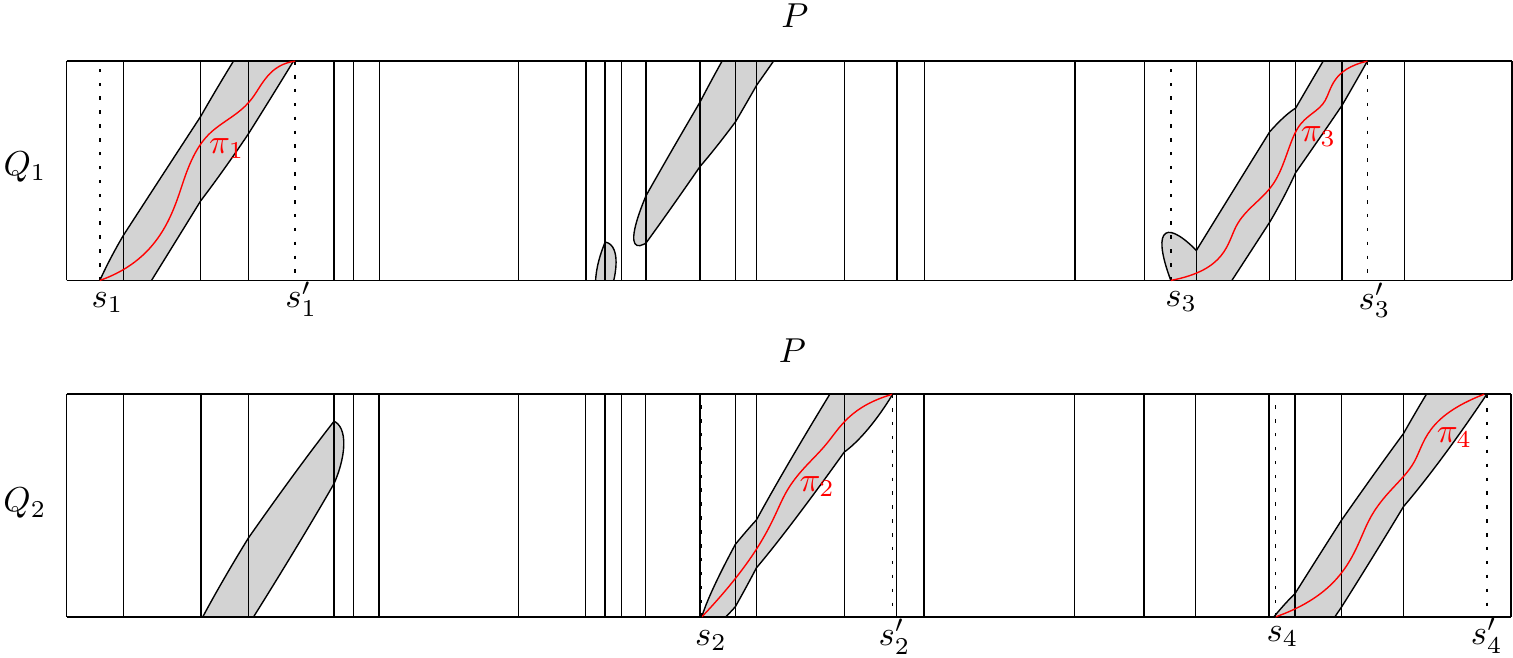}
    \caption{Free space diagrams of the curves $P$ and $Q_1$ (resp. $Q_2$) depicted in Figure~\ref{fig:cover}. The monotone pathes $\pi_i$ illustrate that the Fréchet distance between $P[s_i,s_i']$ and $Q_1$ (resp. $Q_2$) is equal to $\Delta$ for $1\leq i\leq 4$. }
    \label{fig:cover_freespace}
\end{figure}

Our algorithm and analysis use the notion of the free space diagram which was first introduced by Alt and Godau~\cite{AltG95} in an algorithm for computing the Fr\'echet distance. It is instructive to consider this concept in the context of the coverage problem.

\begin{definition}[Free space diagram]
	Let $P$ and $Q$ be two polygonal curves parametrized over $[0,1]$.
	The free space diagram of $P$ and $Q$ is the joint parametric space $[0,1]^2$ together with a not necessarily uniform grid, where each vertical line corresponds to a vertex of $P$ and each horizontal line to a vertex of $Q$.
	The $\Delta$-\emph{free space} of $P$ and $Q$ is defined as \[ \dfree{}{}(P,Q) = \left\{(x,y)\in[0,1]^2 \mid \|P(x) -Q(y)\|\leq\Delta \right\} \] 
	This is the set of points in the parametric space, whose corresponding points on $P$ and $Q$ are at a distance at most $\Delta$.
	The edges of $P$ and $Q$ segment the free space into cells.
	We call the intersection of $\dfree{}{}(P,Q)$ with the boundary of cells the \emph{free space intervals}.
\end{definition}%

Alt and Godau~\cite{AltG95} showed that the $\Delta$-free space inside any cell is convex and has constant complexity.
More precisely, it is an ellipse intersected with the cell.
Furthermore, the Fréchet distance between two curves is less than or equal to $\Delta$ if and only if there exists a path $\pi:[0,1]\rightarrow\dfree{}{}(P,Q)$ that starts at $(0,0)$, ends in $(1,1)$ and is monotone in both coordinates. 

We define the notion of partial traversal, which is a path inside the free space diagram, and relate it to the coverage problem. This will help us to define a finite set of candidates in the next section.

\begin{definition}[Partial traversal]
    Let $P$ be a polygonal curve in $\RR^d$, and let $(t_1,\ldots,t_n)$ be the vertex-parameters of $P$. Let $1 \leq i < j \leq n$ be integer values. Let $Q$ be an edge in $\RR^d$. We define an $(i,j)$-\emph{partial traversal} as a pair of continuous, monotone increasing and surjective functions, $f:[0,1]\rightarrow[a,b]$ and $g:[0,1]\rightarrow[c,d]$, where $t_i\leq a\leq t_{i+1}$,  $t_{j-1}\leq b\leq t_{j}$, $0 \leq a \leq b \leq 1$, and $0\leq c\leq d\leq 1$. We say that $(f,g)$ is a partial traversal from $(a,c)$ to $(b,d)$. 
    We say that a partial traversal is $\Delta$-\emph{feasible} if the image of the path $\pi:[0,1] \rightarrow [0,1]^2$ defined by $\pi(t)=(f(t),g(t))$ is contained inside the $\Delta$-free space $\dfree{}{}(P,Q)$. 
    We say that $\pi$ \emph{covers} a point $t$ on $P$ if $t\in[a,b]$  and  we say that $\pi$ covers a point $t$ on $Q$ if $t \in [c,d]$.   
\end{definition}

\begin{figure}[t]
    \centering
    \includegraphics[width=0.8\textwidth]{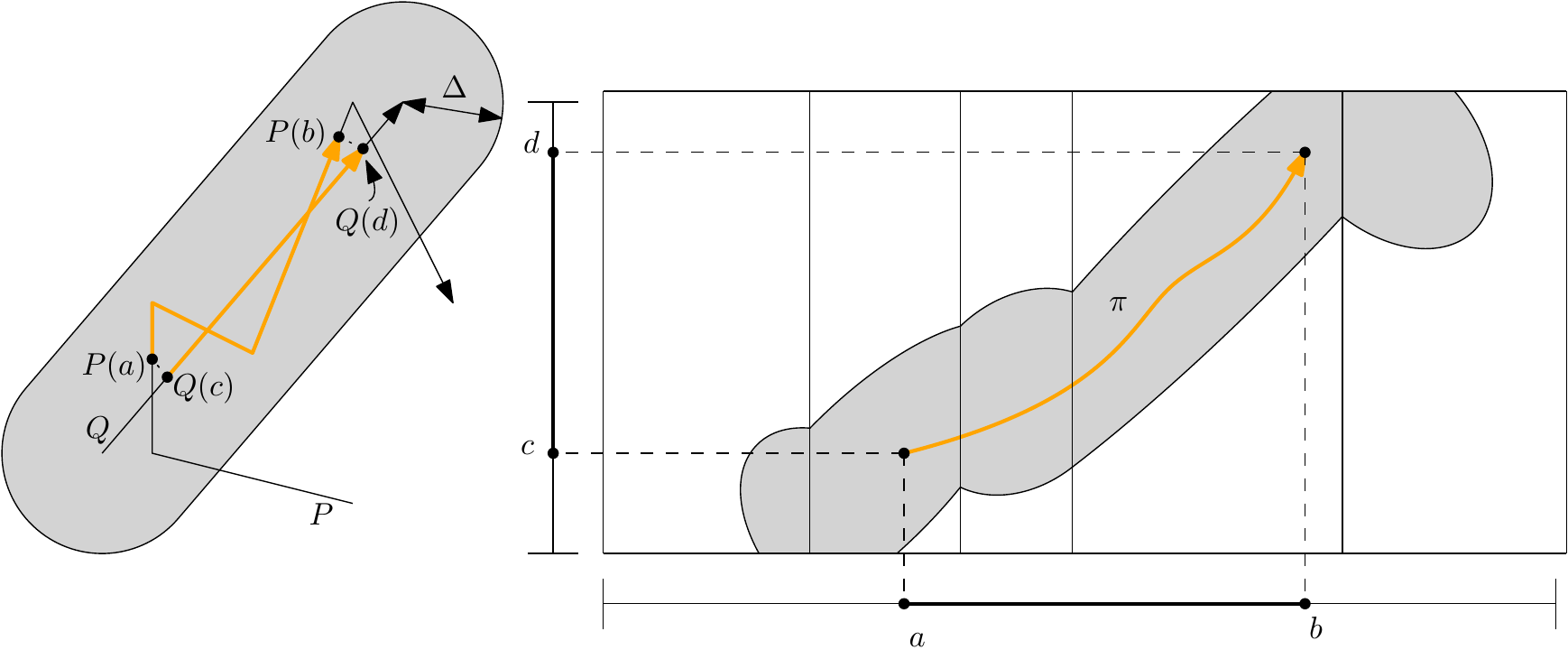}
    \caption{An illustration of a $\Delta$-feasible $(2,4)$-partial traversal $\pi$ from $(a,c)$ to $(b,d)$ of $P$ and $Q$. $\pi$ covers all points between $a$ and $b$ on $P$, and all points between $c$ and $d$ on $Q$.}
    \label{fig:coverage_radius}
\end{figure}

\subsection{A finite set of candidates}
\label{sec:alg_candidates}

By Theorem~\ref{thm:candidate_space}, it is sufficient to find a structured covering using a suitable simplification of the input curve. 
However the corresponding search space would still be infinite, even for a single edge. In this section, we will define a finite set of candidates and show that it contains a good solution.
In particular, our goal is to prove the following theorem.

\begin{theorem}\label{thm:alg_candidates}
     Let $P$ be a polygonal curve of complexity $n$ in $\RR^d$ for some $\Delta>0$. Let $S$ be a $\Delta$-good simplification of $P$. Assume there exists a $\Delta$-covering $C$ of $P$ of cardinality $k$.
    Then, there exists an algorithm that computes in $O(n^3)$ time and space a set of candidates $B \subset\CandidatesArg{E(S)}\subset \XX^d_2$ with $|B| \in O(n^3)$, such that $B$ contains a structured $8\Delta$-covering $C_B$ of $S$ of size at most $12k$. Moreover, $C$ is a $11\Delta$-covering of $P$. 
\end{theorem}

The main steps to constructing this set of candidates $B$ are as follows. We first define a special set of subcurves of the simplification $S$. Intuitively, these are the containers of $S$ of subcurves of $P$ that may contribute to the coverage.

\begin{definition}[Generating subcurves]
    Let $S$ be a $\Delta$-good simplification of a polygonal curve $P$. Let $(t_1,\ldots,t_m)$ be the vertex-parameters of $S$. For any $1\leq i$, $1\leq j\leq 4$ and $i+j\leq m$, we say the subcurve
    $S[t_i,t_{i+j}]$ is a \emph{generating subcurve}. In particular, this defines all subcurves of at most three edges starting and ending at vertices of $S$. 
\end{definition}

\begin{remark}
We remark that it may be possible to work directly on $P$, instead of $S$, by defining generating subcurves on $P$. However, this would likely lead to a higher number of candidates, since we potentially would have to consider  $O(n^2)$ generating subcurves. Using the definitions above, we obtain at most $O(n)$ generating subcurves on $S$ and, as a result, we will get an upper bound of $O(n^3)$ on the total size of the candidate set. Moreover, we will show in Section~\ref{sec:analysis:cpacked} that, using this definition, the  bound on the size of the  candidate set can be  improved even further if $P$ is $c$-packed.
\end{remark}

Now, for every generating subcurve $Y$ of $S$ and every edge $e$ of $S$, we can identify an interval on $e$, that maximizes the $\Delta$-coverage on $Y$ over all subedges of $e$, for the exact statement refer to Corollary~\ref{cor:pathcreator}. For this reason, we call the endpoints of this interval extremal.

\begin{figure}[t]
    \centering
    \includegraphics[width=0.9\textwidth]{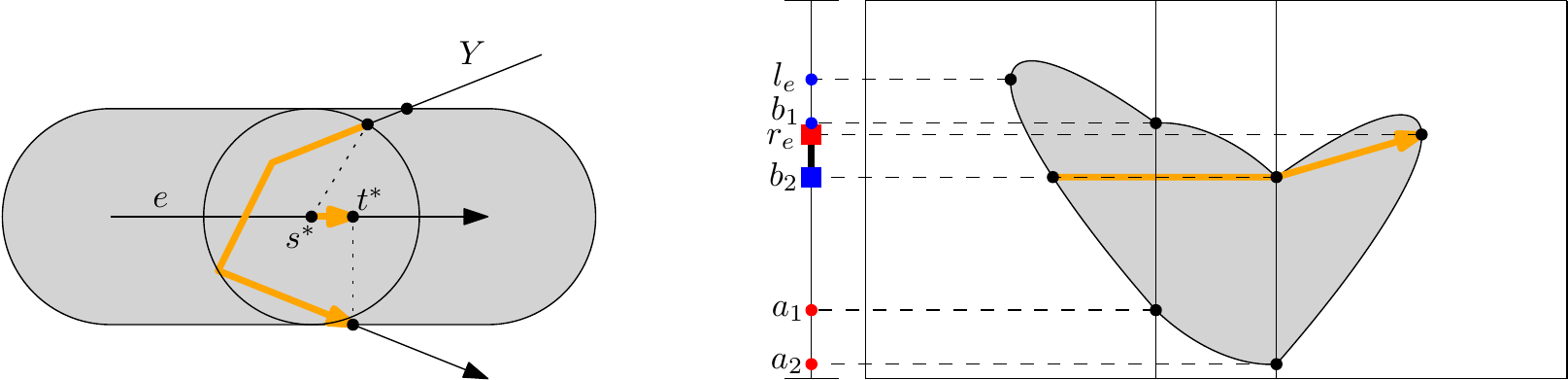}
    \caption{Examples of extremal points. Shown on the right are the two free space intervals $[a_1,b_1]$ and $[a_2,b_2]$ as well as the left- and rightmost points $l_e$ and $r_e$ of the $\Delta$-free space of $e$ and $Y$. The extremal points are defined by $b_2$ and $r_e$. All points considered for the first extremal point are shown in blue. Similarly all points considered for the second extremal point are shown in red. A traversal from the first extremal point to the second extremal point is illustrated. On the left the resulting subedge $e[s^*,t^*]$ and the maximal subcurve of $Y$ that can be matched are illustrated.}
    \label{fig:canonicalpoints}
\end{figure}

\begin{definition}[$\Delta$-extremal points]
    Given a value of $\Delta > 0$, a polygonal curve $Y:[0,1]\rightarrow\RR^d$ of $m$ edges and an edge $e:[0,1]\rightarrow\RR^d$, such that they permit a $\Delta$-feasible $(1,m)$-partial traversal. As $e$ is a single edge, the $\Delta$-free space of $Y$ and $e$ consists of a single row. Let $[a_i,b_i]$ be the $i$th vertical free space interval of the $\Delta$-free space of $Y$ and $e$. Denote by $l=(l_Y,l_e)$ the leftmost point in the $\Delta$-free space of $Y$ and $e$ and $r=(r_Y,r_e)$ the rightmost point (in case $l$ is not unique, chose the point with smallest $y$-coordinate, and $r$ as the point with the biggest $y$-coordinate). We define the $\Delta$-\emph{extremal points} induced by $Y$ on $e$ as the tuple $\mathcal{E}_{\Delta}(Y,e)=(s,t) \in [0,1]^2$ with 
    $s = \min(\{l_e\}\cup\{b_1,\ldots,b_{n-1}\})$ and $t=\max(\{r_e\}\cup\{a_1,\ldots,a_{n-1}\})$. Refer to Figure~\ref{fig:canonicalpoints} We call $s$ the first and $t$ the second extremal point. We explicitly allow that $t < s$. For this special case refer to Figure~\ref{fig:canonicalpoints2}
\end{definition}

Now we are ready to define the finite candidate set induced by $S$. For this, we need the definition of generating triples.

\begin{definition}[Generating triples]\label{def:triples}
Let $S$ be a $\Delta$-good simplification of a polygonal curve $P$. We define the set of \emph{generating triples} $T_S$ as a set of triples $(e,Y_1,Y_2)$, where $e$ is any edge of $S$, and $Y_1$ and $Y_2$ are generating subcurves of $S$ (not necessarily distinct). We include the triple $(e,Y_1,Y_2)$ in the set $T_S$ if and only if there are points $p\in e$, $p_1\in Y_1$ and $p_2\in Y_2$ such that $||p-p_1||\leq8\Delta$ and $||p-p_2||\leq8\Delta$.
\end{definition}

\begin{definition}[Candidate set]\label{def:candidates}
Let $\Delta>0$ be a given value and let $S$ be a $\Delta$-good simplification of a polygonal curve $P$. Let $T_S$ be the set of generating triples of $S$. We define the \emph{candidate set} induced by $S$ with respect to $\Delta$ as the set of line segments
\[B = \{ e[s_1,t_2] \mid \exists~ (e,S_1,S_2) \in T_S, \text{ s.t. } \mathcal{E}_{8\Delta}(S_i,e)=(s_i,t_i) \text{ for } i \in \{1,2\} \}\]
Clearly, the set $B$ can be computed in $O(|T_S|)$ time and space, if the set $T_S$ is given.
\end{definition}

\begin{figure}[t]
    \centering
    \includegraphics[width=0.8\textwidth]{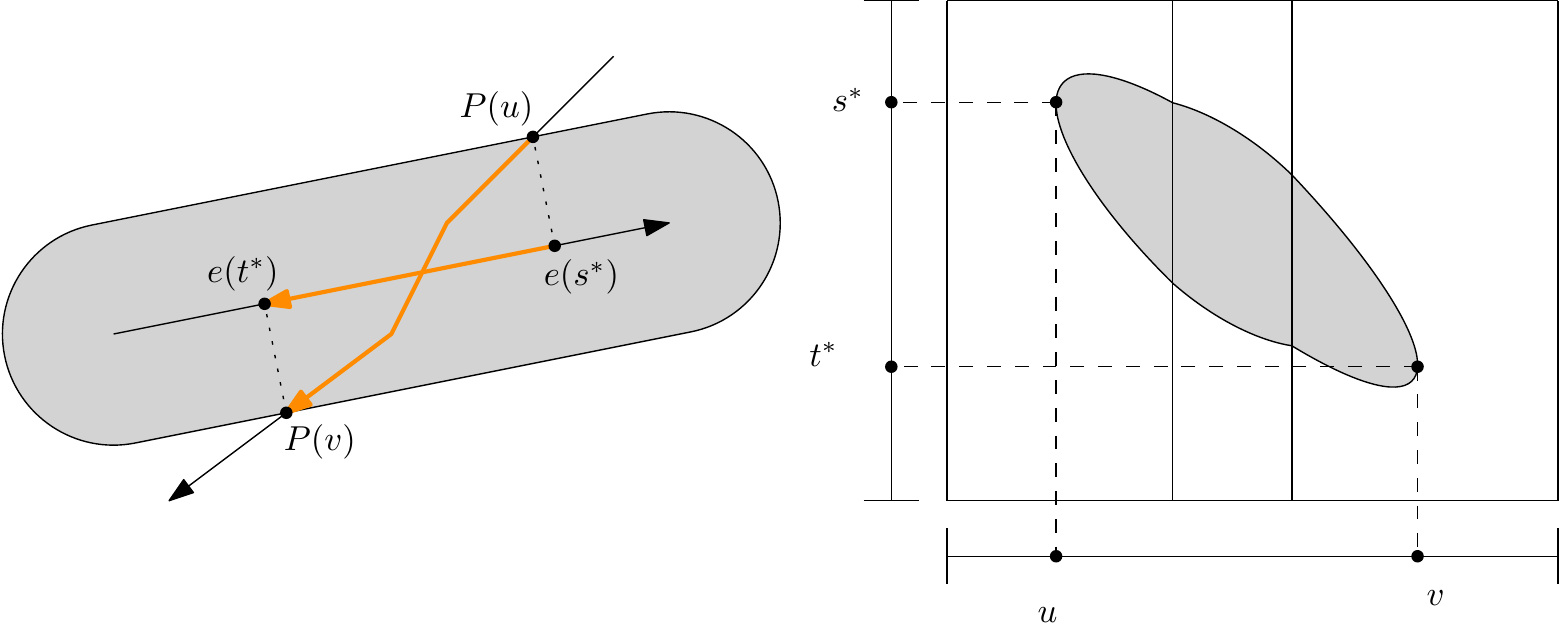}
    \caption{Illustration of $\Delta$-extremal points of a polygoncal curve $P$ and an edge $e$, where the first extremal point lies above the second. The two induced subcurves that permit a $\Delta$-feasible partial traversal covering all points on $e[s^*,t^*]$ are shown in orange. Note that $e[s^*,t^*]$ does not point in the same direction as $e$.}
    \label{fig:canonicalpoints2}
\end{figure}

\subsection{Analysis and proofs}\label{sec:proofssec3}

We now want to prove Theorems~\ref{thm:candidate_space} and \ref{thm:alg_candidates}. We will use the following observation on the Fr\'echet distance of a curve and its simplifications.

\begin{observation}\label{obs:dfsimp}
    Let $P$ be a polygonal curve, and let $S$ be a $\Delta$-good simplification of $P$, defined by the index set $I=(i_1,\ldots,i_k)$. Then $d_F(P,S)\leq3\Delta$. Moreover, there is a traversal $(f_S,g_S)$ with $0\leq t_1\leq\ldots\leq t_k\leq1$, such that $P(f_S(t_j))=S(g_S(t_j))=p_{i_j}$, with associated distance at most $3\Delta$. This can be seen by concatenating the traversals induced by conditions $(ii)$ and $(iii)$ on the respective subcurves.
\end{observation}

The following Lemma motivates the use of the simplification $S$. It shows that for any covering of $P$ there exists a suitable structured covering of $S$. Moreover, we can transfer a structured cover of $S$ back to $P$.

\begin{lemma}\label{lem:covertransfer}
    Let $P$ be polygonal curve, and let $S$ be a $\Delta$-good simplification of $P$. Let $P'$ be a polygonal curve, with $d_F(P,P')\leq\Delta'$. Assume there exists a set $C\subset\XX^d_2$ of cardinality $k$, such that $\Psi_\Delta(P,C) = [0,1]$. Then 
    \begin{itemize}
        \item[(i)] $\Psi_{\Delta+\Delta'}(P',C) = [0,1]$ and
        \item[(ii)] $\Psi'_{4\Delta}(S,C) = [0,1].$
    \end{itemize} 
\end{lemma}
\begin{proof}
    We start by proving $(i)$.
    Let $(f,g)$ be a traversal of $P$ and $P'$, with associated cost at most $\Delta'$. Let $\mu_P(x) = \{y\in[0,1]\mid\exists t\in[0,1]:f(t)=x,\,\, g(t)=y\}$, that is all the (parametrized) points along $P'$, that get matched to $P(x)$ during some traversal with associated distance at most $\Delta'$. Note that $\mu_P([0,1]) = [0,1]$, and more importantly for $[a,b]\subset[c,d]$, it holds that $\mu_P([a,b]) \subset \mu_P([c,d])$, as $f$ and $g$ are monotone. 

    We claim that
    \begin{eqnarray*}
    [0,1]= \mu_P([0,1]) = \mu_P\left( \Psi_{\Delta}(P,C) \right) \subseteq \Psi_{\Delta+\Delta'}(P',C)  
    \end{eqnarray*}
    This would imply the set inclusion $[0,1] \subseteq \Psi_{\Delta+\Delta'}(P',C)$, which then also implies equality, since by definition $[0,1] \supseteq \Psi_{\Delta+\Delta'}(P',C)$.
    
    We argue as follows. Observe that by triangle inequality it holds for any $t,t' \in [0,1]$ and any $Q \in C$ with $d_F(P[t,t'],Q) \leq \Delta$, and for any 
    $s \in \mu_P(t)$ and $s'\in \mu_P(t')$ that 
    \begin{eqnarray*}
    d_F(P'[s,s'], Q) \leq d_F(P'[s,s'], P[t,t']) + d_F(P[t,t'],Q) \leq \Delta+\Delta' 
    \end{eqnarray*}

    Therefore we can write
    \begin{eqnarray*}
    \mu_P\left( \Psi_{\Delta}(P,C) \right) &=& \bigcup_{Q\in C} \bigcup_{0 \leq t \leq t'\leq 1} \{ x \in \mu_P([t,t']) \mid d_F(P[t,t'],Q) \leq \Delta\}\\
    &\subseteq&  \bigcup_{Q\in C} \bigcup_{0 \leq s \leq s'\leq 1} \{ x \in [s,s'] \mid d_F(P'[s,s'],Q) \leq \Delta + \Delta'\}\\
    &=& \Psi_{\Delta+\Delta'}(S,C) 
    \end{eqnarray*}
    Indeed, the second step follows from the above observation since $\mu_P([t,t'])=[s,s']$ for some $s \in \mu_P(t)$ and $s' \in \mu_P(t')$ with $s \leq s'$ since $f$ and $g$  are monotone.
    
    Now for $(ii)$ notice, that when $P'=S$ is a $\Delta$-good simplification of $P$, and $\Delta'=3\Delta$, $S[s,s']$ is contained in $c_S(P[t,t'])$ for $\mu_P([t,t'])=[s,s']$. This follows from the traversal given in Observation \ref{obs:dfsimp} together with the definition of $\mu_P$. Thus, because $c_S(P[t,t'])$ consists of at most three edges by Lemma \ref{lem:three:edges}, it holds that 
    \[\bigcup_{Q\in C} \bigcup_{0 \leq s \leq s'\leq 1} \{ x \in [s,s'] \mid d_F(S[s,s'],Q) \leq 4\Delta\} \subset \Psi'_{4\Delta}(S,C),\]
    implying the claim.
\end{proof}

The following lemma shows that we can restrict the search for a covering to the subedges of the simplification.

\begin{lemma}\label{thm:approx}
Let $S$ be a  $\Delta$-good simplification of some curve $P$. Assume there exists a set $C \subset \XX^d_{2}$ of size $k$, such that $\Psi_{\Delta}(P,C)= [0,1] $. Then there exists a set $\ConcreteCandidates{S} \subset \CandidatesArg{E(S)}$ of cardinality at most $3k$ such that $\Psi'_{8\Delta}(S,\ConcreteCandidates{S})= [0,1]$.
\end{lemma}
\begin{proof}
    We start by applying Lemma \ref{lem:covertransfer} $(ii)$ to $P$ and its $\Delta$-good simplification. Thus $\Psi'_{4\Delta}(S,C) = [0,1]$.
    We show that for each center curve in $Q\in C$ there exist 3 subcurves of edges of $S$ that cover all the parts of $S$ that were covered by $Q$. An illustration of the proof is given in Figure~\ref{fig:cov_approx}.
    Let $Q\in C$ with $\Psi'_{4\Delta}(S,\{Q\})\neq\emptyset$. We fix a subcurve $\pi$ of $S$ such that $d_F(\pi,Q)\leq 4\Delta$, that consists of at most $3$ edges $e_1$, $e_2$, $e_3\in\CandidatesArg{E(S)}$. $\pi$ exists by Lemma \ref{lem:three:edges}.
    The curve $Q$ can be split into $3$ subcurves  $Q_1$, $Q_2$, $Q_3$ such that $Q=Q_1\oplus Q_2\oplus Q_3$ with $d_F(Q_1,e_1)\leq 4\Delta$, $d_F(Q_2,e_2)\leq 4\Delta$ and $d_F(Q_3,e_3)\leq 4\Delta$.
    
    \begin{figure}
        \centering
        \includegraphics[width=0.65\textwidth]{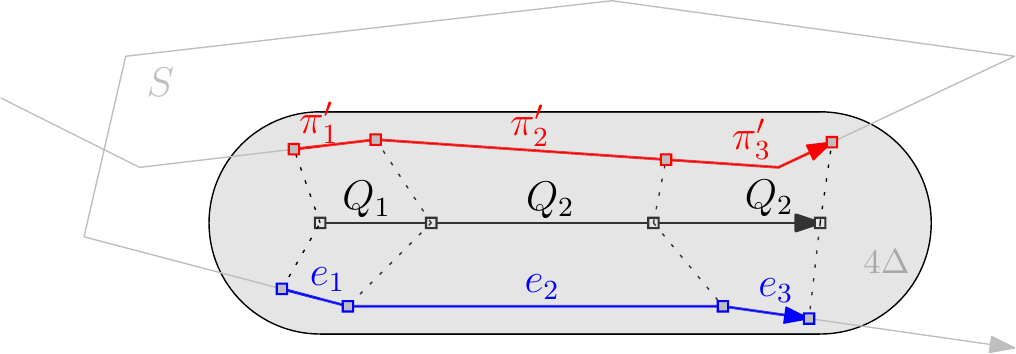}
        \caption{Illustration to the proof of Lemma~\ref{thm:approx}. The example illustrates that the $8\Delta$-coverage of the edges $e_1$, $e_2$ and $e_3$ is a superset of the $4\Delta$-coverage of $Q=Q_1\oplus Q_2\oplus Q_3$. }
        \label{fig:cov_approx}
    \end{figure}
    Now consider an arbitrary subcurve $\pi'$ of $S$ such that $d_F(\pi',Q)\leq 4\Delta$. The curve $\pi'$ can be split into $3$ subcurves  $\pi'_1$, $\pi'_2$, $\pi'_3$ such that $\pi'=\pi'_1\oplus\pi'_2\oplus\pi'_3$ with $d_F(\pi'_1,Q_1)\leq 4\Delta$, $d_F(\pi'_2,Q_2)\leq 4\Delta$ and $d_F(\pi'_3,Q_3)\leq 4\Delta$. By triangle inequality we get
    \[d_F(\pi'_1,e_1)\leq d_F(\pi'_1,Q_1)+d_F(Q_1,e_1)\leq 8\Delta. \]
    In the same way we obtain $d_F(\pi'_2,e_2)\leq 8\Delta $ and $d_F(\pi'_3,e_3)\leq 8\Delta$. It follows that the entire curve $\pi'$ is covered by $e_1$, $e_2$ and $e_3$. By applying this argument to every subcurve $\pi'$ of $S$ with $d_F(\pi',Q)\leq 4\Delta$ with $\pi'$ consising of at most three edges,  we get
    \[\Psi'_{4\Delta}(S,\{Q\})\subseteq \Psi'_{8\Delta}(S,\{e_1,e_2,e_3\}).\]
    Applying this to every $Q \in C$ and using the fact that $\Psi'_{8\Delta}(S,C) = \bigcup_{Q \in C} \Psi'_{8\Delta} (S,\{Q\})$, the lemma is implied by constructing the set $\ConcreteCandidates{S}$ out of the edges $e_1,e_2,e_3$ for each $Q \in C$.
\end{proof}

Using the above two lemmas, we can prove Theorem~\ref{thm:candidate_space}, which was the main theorem in Section~\ref{sec:structuredCoverage}.

\begin{proof}[Proof of Theorem~\ref{thm:candidate_space}]
     Lemma~\ref{thm:approx}  implies that there exists a set $C$ of subedges of edges of $S$ of size $3k$ which is a structured $8\Delta$-covering of $S$.
     By Observation \ref{obs:unstructure_cover},
    $C$ is also an $8\Delta$-covering of $S$. 
    Now, Observation~\ref{obs:dfsimp}  implies that $d_F(P,S)\leq3\Delta$ and thus we can apply Lemma~\ref{lem:covertransfer}~$(i)$ with $P'=S$ and $\Delta'=3\Delta$ to conclude that $C$ is an $11\Delta$-covering of $P$.     
\end{proof}

In the remainder of the section, we want to prove Theorem~\ref{thm:alg_candidates} of Section~\ref{sec:alg_candidates}. In particular, we want to prove that our construction of the candidate set (Definition~\ref{def:candidates}) satisfies the needs of this theorem. 
The idea is as follows. We now look at any structured $8\Delta$-covering of the simplification $S$ consisting of only subedges of $S$. We want to deform each such subedge to one of our candidates. 
Lemma~\ref{lem:pathcreator} below shows that we can continuously deform any subedge  to some edge from our candidate set while retaining the coverage on a subcurve $Y$ of $S$. In particular, any deformation that is monotone has this property. 
However, this deformation is specific to a single subcurve covered by this subedge. Thus, while retaining coverage on one subcurve, we may lose coverage on another subcurve in the same cluster. Lemma~\ref{lem:settransform} will show how to deal with all subcurves at once while increasing the number of clusters by a factor of at most $4$.

\begin{restatable}{lemma}{pathcreator}
\label{lem:pathcreator}
    Let $\mathcal{E}_{\Delta}(Y,e)=(s^*,t^*)$ be the extremal points for some $Y$, $e$ and $\Delta$. Let the $Y$ consist if $m$ edges.
    Let further $0\leq s \leq t \leq 1$ be given. Then for all values $s'$ between $s$ and $s^*$, and $t'$ between $t$ and $t^*$ we have that
    $$\Psi^{(1,m)}_\Delta(Y,e[s,t])\subset \Psi^{(1,m)}_\Delta(Y,e[s',t']).$$
\end{restatable}

\begin{proof}

    Let $s'$ and $t'$ be fixed but arbitrary values between $s$ and $s^*$, and $t$ and $t^*$.Let $(t_1,\ldots,t_{m+1})$ be the vertex parameters of $Y$.
    Let $0\leq u,v\leq 1$ be two arbitrary values such that $u<t_2$ and $u<v$ and $t_{m}<v$, such that $d_F(Y[u,v],e[s,t])\leq \Delta$.
    
    By Definition \ref{def:structuredcoverage}, it is enough to show the existence of two values $0\leq u',v'\leq 1$ such that $u'<t_2$ and $u'<v'$ and $t_{m}<v'$ and $d_F(Y[u',v'],e[s',t'])\leq \Delta$ together with $[u,v]\subset[u',v']$.
    
    Denote by $\pi:[0,1]\rightarrow[0,1]^2$ a monotone path from $(u,s)$ to $(v,t)$ in the $\Delta$-free space of $Y$ and $e$. $\pi$ is induced by the traversal of $Y[u,v]$ and $e[s,t]$. Denote by $\mathcal{D}^i(Y,e)$ the intersection $\dfree{}{}(Y,e)\cap C_{i}$ where $C_{i}$ denotes the cell in the $\Delta$-free space of $Y$ and $e$ corresponding to the $i$th edge of $Y$ and $e$.
    Let $[a_i,b_i]$ be the $i$ vertical free space interval. Note that these are nonempty, as there is a $(1,m)$-partial traversal of $Y$ and $e$.
    
    Note, that $s^*$ is the first $\Delta$-extremal point of $Y$ on $e$, which is by definition the $y$-coordinate of the lowest leftmost point in $\mathcal{D}^1(Y,e)$ from which traversals into the last cell are possible. Call this lowest leftmost point $l$. Similarly define $r$ as the highest rightmost point, that is the point realizing $t^*$.
    
    Now let $s'$ and $t'$ be given, as some value between $s$ and $s^*$, and $t$ and $t^*$ respectively. Denote the leftmost point with $y$-coordinate $s'$ by $l'=(u',s')$ and similarly the rightmost point with $y$-coordinate $t'$ by $r'=(v',t')$ defining $u'$ and $v'$.
    
    Note that by convexity of $\mathcal{D}^1(Y,e)$ and $\mathcal{D}^{n-1}(Y,e)$, $l'$ lies to the left of $(u,s)$, and $r'$ lies to the right of $(v,t)$, as $l$ and $r$ do so too, thus $[u,v]\subset[u',v']$.
    
    Now, we need to construct a monotone path $\pi'$ from $l'$ to $r'$ which proves the claim. Refer to Figure \ref{fig:pathcreator} for illustrations of four of the following seven cases.
    
    \begin{compactitem}
    \item \textbf{Case 1:} $s\geq s^*$ and $t\leq t^*$. But then $l'$ lies to the left and below $(u,s)$, and $r'$ above and to the right of $(v,t)$. Thus $\overline{l'\,(u,s)}\oplus \pi \oplus \overline{(v,t)\,r'}$ is the sought after path.
    \item \textbf{Case 2:} $s\leq s^*$ and $t\leq t^*$. Hence $s'\leq s^*$, and thus $s'\leq \min({b_1,\ldots,b_{n-1}})$. 
        \begin{compactitem}
            \item  \textbf{Case 2a:} $s'\leq t$. As $\pi$ starts below $l'$ (at $(u,s)$), we walk rightwards from $l'$, until we intersect $\pi$, and follow $\pi$ until we reach the last cell at point $p$, that is $\dfree[]{1}{n-1}(Y,e)$. Note that the intersection exists, as $s\leq s'\leq t$, and $\pi$ being monotone. As all points must lie below $t$, and by extension $r'$, $p$ lies below $r'$. Thus we end $\pi'$ with a straight line from $p$ to $r'$.
            \item \textbf{Case 2b:} $s'>t$. In this case, walking rightwards from $l'$ we will never intersect $\pi$, but we will enter the last cell at point $p$, call this path $\pi'$. Now $r'$ may lie below $p$. This however is no problem, as so far $M\circ\pi'$ is also a $(1,n)$-partial traversal of $Y$ and $e[1,0]$, for $M:[0,1]^2\rightarrow[0,1]^2;(x,y) \mapsto (x,1-y)$ which mirrors the $y$-coordinate. Thus we append a straight line from $p$ to $r'$ resulting in the soughtafter path $\pi'\oplus \overline{p\,r'}$ for $e$ or $(M\circ\pi')\oplus\overline{M(p)\,M(r')}$ for $e[1,0]$. Since $(e[1,0])[1-s',1-t'] = e[s',t']$ the claim follows.
        \end{compactitem}
   \item \textbf{Case 3:} $s\geq s^*$ and $t\geq t^*$. This is symmetric to Case 2, by starting a leftwards walk from $r'$ until we intersect $\pi$, and then connecting it to $l'$.
   \item \textbf{Case 4:} $s\leq s^*$ and $t\geq t^*$. Hence $s'\leq s^*$, and thus $s'\leq \min({b_1,\ldots,b_{n-1}})$, and similarly $t'\geq \max({a_1,\ldots,a_{n-1}})$.
   \begin{compactitem}
    \item \textbf{Case 4a:} $s'\leq t'$. We walk rightwards from $l'$ until we intersect $\pi$ at $p$ and similarly walk leftwards from $r'$ until we intersect $\pi$ at $q$. as $p$ lies below $q$, the path is given by $\overline{l'\,p}\oplus\pi[p,q]\oplus\overline{p\,r'}$.
    \item \textbf{Case 4b:} $s'\geq t'$. This case is symmetric to Case 2a. We walk rightwards from $l'$ until we reach the last cell, and then connect $r'$ to this path, resulting in a monotone path in $\dfree[]{}{}(Y,e[1,0])$. The rightwards walk may intersect $\pi$, but this is no problem, as all $a_i$ lie below $s'$, and all $b_i$ lie above $s'$ by assumption. 
   \end{compactitem}
    \end{compactitem}
    
    Thus the claim follows.
\end{proof}

\begin{corollary}\label{cor:pathcreator}

 Let $\mathcal{E}_{\Delta}(Y,e)=(s^*,t^*)$ be the extremal points for some $Y$, $e$ and $\Delta$. For any $0\leq s \leq t \leq 1$ it holds that
    $$\Psi^{(1,m)}_\Delta(Y,e[s,t])\subset \Psi^{(1,m)}_\Delta(Y,e[s^*,t^*]).$$
\end{corollary}

\begin{definition}[inclusion-minimal generating subcurve]
Let $\pi=S[s,t]$ be a subcurve of some polygonal curve $S$, such that $\pi$ consists of at most three edges. Let $(t_1,\ldots,t_m)$ be the vertex-parameter of $S$, and $v_i=S(t_i)$ the vertices of $S$. Let $S(s)$ lie on edge $e_i=\overline{v_i\,v_{i+1}}$ and $S(t)$ on $e_j=\overline{v_j\,v_{j+1}}$. Define the inclusion-minimal generating subcurve containing $\pi$ as the subcurve $S[t_i,t_{j+1}]$.
\end{definition}

Note that for every $\pi$ consisting of at most three edges, the inclusion-minimal generating subcurve containing $\pi$ is indeed a generating subcurve of $S$, as $\pi$ has at most three edges.
Note further, that these inclusion-minimal generating subcurves do not correspond to inclusion-minimal subcurves as defined by de Berg et al.\ in $\cite{de2013fast}$.

\begin{figure}
    \centering
    \includegraphics[width=\textwidth]{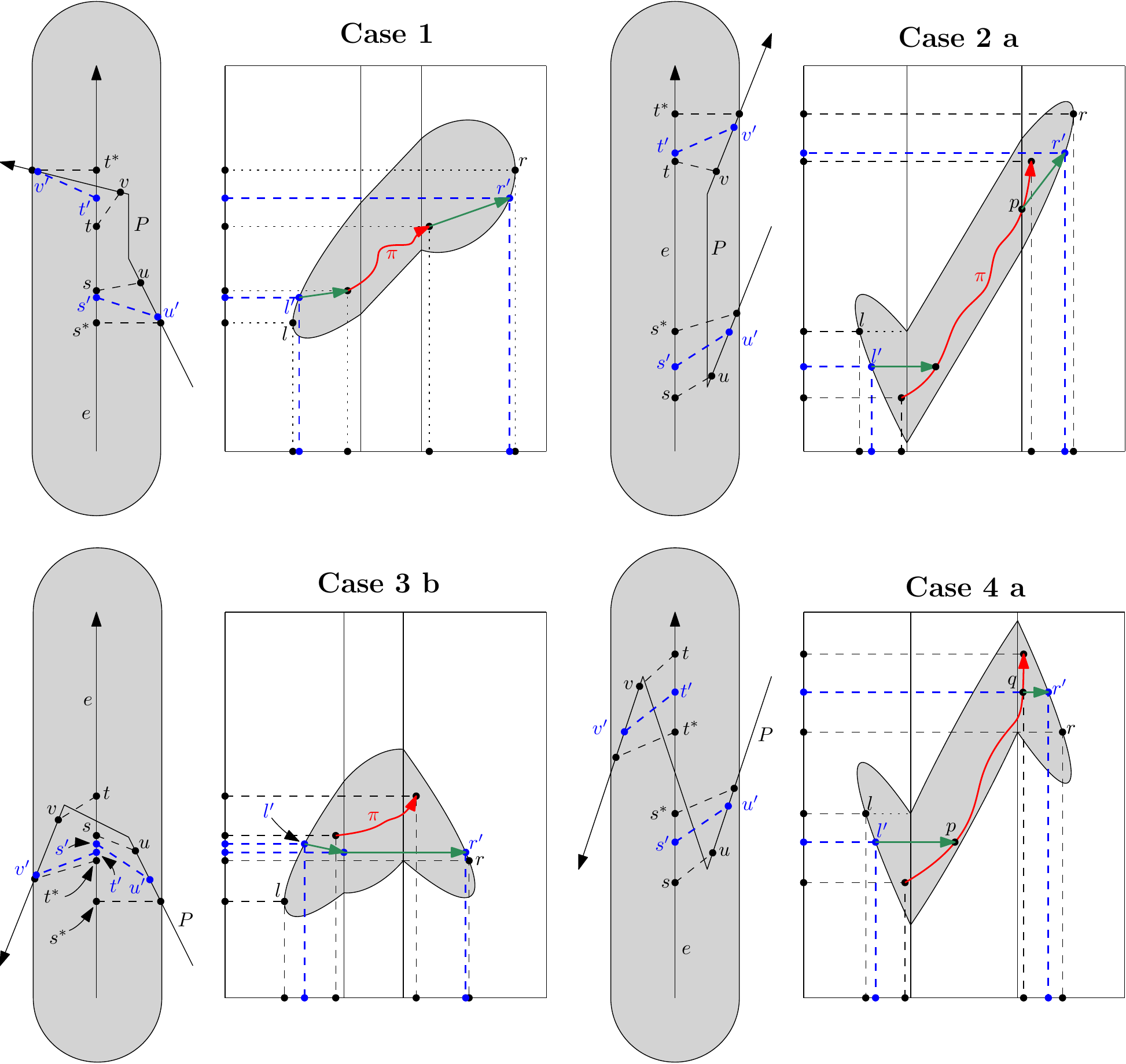}
    \caption{Illustration to the case analysis in the proof of Lemma \ref{lem:pathcreator}.  The figure shows a subset of crucial cases.   \textbf{Case 1:} $s\geq s^*$ and $t\leq t^*$. \textbf{Case 2a:} $s\leq s^*$, $t\leq t^*$ and $s'\leq t$. \textbf{Case 3b:} $s\geq s^*$, $t\geq t^*$ and $s\geq t'$. \textbf{Case 4:} $s\leq s^*$, $t\geq t^*$ and $s'\geq t'$.}
    \label{fig:pathcreator}
\end{figure}

\begin{restatable}{lemma}{canonicalsettransform}
\label{lem:settransform}
    Let $\Delta>0$ be a given value and let $S$ be a $\Delta$-good simplification of some polygonal curve $P$. Let $B$ be the corresponding candidate set (Definition~\ref{def:candidates}).   
    Let $E$ be the set of edges of $S$. 
    For any edge of the candidate space, $x \in \CandidatesArg{E}$, there is a set $C' \subseteq B$ with $|C'| \leq 4$, such that 
    $\Psi'_{8\Delta}(S,\{x\})\subseteq \Psi'_{8\Delta}(S,C')$.
    In other words, we can replace $x$ by a specific subset of at most four edges from the candidate set and still retain the structured coverage on~$S$. 
\end{restatable}

\begin{proof}
    Let $e$ be the edge of $S$ that contains $x$ and, in particular, let $0\leq s \leq t \leq 1$, such that $x=e[s,t]$.
    Let $\Pi = \{\pi=S[s_\pi,t_\pi] \mid d_F(\pi,e[s,t])\leq8\Delta,\text{$\pi$ consists of at most $3$ edges}\}$ be the possibly infinitely big set of subcurves of $S$, that define $\Psi'_{8\Delta}(S,e[s,t])$.
    More precisely $\bigcup_{\pi\in\Pi}[s_\pi,t_\pi]=\Psi'_{8\Delta}(S,e[s,t])$.
    Let $C_\Pi$ be the set of inclusion-minimal generating subcurves of $S$ containing some $\pi\in\Pi$. Note that $|C_\Pi|=O(n^3)$. Partition $C_\Pi$ into $\mathcal{G}=\{G_{<\geq},G_{\geq\geq},G_{<<},G_{\geq<}\}$, where $G_{<\geq}=\{Y\in C_\Pi\mid s_{Y,e}<s \;;\; t\geq t_{Y,e}\}$, where $s_{Y,e}$ and $t_{Y,e}$ are the $(8\Delta)$-extremal points of $Y$ on $e$ and the other elements of $\mathcal{G}$ are defined similarly.
    
    We will handle each set of containers in the partition separately, resulting in  four subedges of $e$ forming the set $C'$, which we will denote by $e_k=e[\alpha_k,\alpha_k]$ for $1 \leq k\leq 4$.
    
    We begin with $G_{<\geq}$. By definition, for every subcurve $\pi=S[s_\pi,t_\pi]\in\Pi$, such that $Y'=c_S(\pi)$ ended up in $G_{<\geq}$, the left $8\Delta$-extremal point of $Y'$ on $e$ lies below $s$. Define $\alpha_1=\max(\{s_{Y,e}\mid Y \in G_{<\geq}\})$. Similarly all right $8\Delta$-extremal points lie above $t$. Define $\beta_1=\min(\{t_{Y,e}\mid Y \in G_{<\geq}\})$. Then for all these subcurves $\alpha_1$ lies between $s$ and the first $8\Delta$-extremal point $s_{Y,e}$, where $Y$ is its container. Similarly $\beta_1$ lies between $t$ and the second $8\Delta$-extremal point $t_{Y,e}$. Thus we can apply Lemma \ref{lem:pathcreator} to every $S[s_\pi,t_\pi]$ together with $\alpha_1$ and $\beta_1$, resulting in values $0\leq s_\pi'\leq s_\pi$ and $t_\pi\leq t'_\pi\leq 1$ such that $d_F(S[s'_\pi,t'_\pi],e[\alpha_1,\beta_1])\leq8\Delta$. Hence 
    \[\bigcup_{\pi\in \Pi \text{ s.t. } c_S(\pi)\in G_{<\geq}}[s_\pi,t_\pi]\subset\bigcup_{\pi\in \Pi \text{ s.t. } c_S(\pi)\in G_{<\geq}}[s'_\pi,t'_\pi]\subset\Psi'_{8\Delta}(S,\{e[\alpha_1,\beta_1]\}).\]
    
    Next handle $G_{<<}$. Define $\alpha_2=\min(\{s_{Y,e}\mid Y \in G_{<<}\})$ and $\beta_2=\min(\{t_{Y,e}\mid Y \in G_{<<}\})$. This definition is close to $\alpha_1$ and $\alpha_2$, except this time all first $8\Delta$-extremal points lie above $s$. Similarly it follows from Lemma \ref{lem:pathcreator} that
    \[\bigcup_{\pi\in \Pi \text{ s.t. } c_S(\pi)\in G_{<<}}[s_\pi,t_\pi]\subset\Psi'_{8\Delta}(S,\{e[\alpha_2,\beta_2]\}).\]
    
    For $G_{\geq\geq}$ define $\alpha_3=\max(\{s_{Y,e}\mid Y \in G_{\geq\geq}\})$ $\beta_3=\max(\{t_{Y,e}\mid Y \in G_{\geq\geq}\})$. and for $G_{\geq<}$ define $\alpha_4 = \min(\{s_{Y,e}\mid Y \in G_{\geq<}\})$ and $\beta_4=\max(\{t_{Y,e}\mid Y \in G_{\geq<}\})$. All together we have
    \begin{align*}
        \Psi'_{8\Delta}(S,e[s,t]) & = \bigcup_{\pi\in\Pi}[s_\pi,t_\pi]= \bigcup_{G\in \mathcal{G}}\left(\bigcup_{\pi\in\Pi\text{ s.t. }c_S(\pi)\in G}[s_\pi,t_\pi]\right)\\
        &\subset \bigcup_{1\leq k\leq 4}\Psi'_{8\Delta}(S,\{e[\alpha_k,\beta_k]\}) = \Psi'_{8\Delta}(S,\{e[\alpha_1,\beta_1],\ldots,e[\alpha_4,\beta_4]\}).
    \end{align*}
    Thus the claim follows.
\end{proof}

We are now ready to prove the main theorem of Section~\ref{sec:structure}. In particular, we want to prove that our candidate set from Definition~\ref{def:candidates} satisfies the properties stated in Theorem~\ref{thm:alg_candidates}.
To this end, we will modify the proof of Theorem~\ref{thm:candidate_space}.

\begin{proof}[Proof of Theorem~\ref{thm:alg_candidates}]
    Let $S$ be a $\Delta$-good simplification of $P$. Lemma~\ref{thm:approx}  implies that there exists a set $C_S$ of subedges of edges of $S$ of cardinality at most $3k$ which is a structured $8\Delta$-covering of $S$.
    Now, by Lemma \ref{lem:settransform}, we can replace each element of $C_S$ with four elements of $B$ and retain the $8\Delta$-coverage on $S$. Thus, there exists a set of candidates 
    $C_B \subseteq B$ of size at most $4|C_S| = 12k$ which constitutes a structured $8\Delta$-covering on $S$.
    At this point, we can simply proceed as in the proof of Theorem~\ref{thm:candidate_space}.    
    By Observation \ref{obs:unstructure_cover},
    $C_B$  is also an $8\Delta$-covering of $S$. 
    Since Observation~\ref{obs:dfsimp}  implies that $d_F(P,S)\leq3\Delta$, we can apply Lemma~\ref{lem:covertransfer}~$(i)$ with $P'=S$ and $\Delta'=3\Delta$ to conclude that $C_B$ is an $11\Delta$-covering of $P$.     
    
    As for the time and space for generating the set $B$, we argue as follows. 
    Note that there are at $O(n)$ generating subcurves on $S$ and at most $n$ edges. Hence there are $O(n^3)$ triples of the form $(e,S_1,S_2)$. For each of these  triples we can check in $O(1)$ time, whether they should be added to the set of generating triples $T_S$. Each element in $T_S$ generates one candidate and computing the $8\Delta$-extremal points for curves of constant complexity can be done in constant time. Thus, we can compute all elements of $B$ in $O(n^3)$ time. As each candidate has constant complexity, the space requirement for storing $B$ also is in $O(n^3)$.
\end{proof}

\section{The main algorithm}
\label{sec:sample} 

We describe the main algorithm in Section~\ref{sec:descr:alg} with pseudocode specified in Algorithm~\ref{alg:main}  and Algorithm~\ref{alg:isfeasible}. Specifications of the missing subroutines are given in Table~\ref{tab:specs}.   
Several building blocks of the algorithm are discussed in Sections \ref{sec:alg_candidates} (Computing candidates),
\ref{sec:alg_coverage} (Computing the structured coverage and testing feasibility), and \ref{sec:simplifications:alg} (Computing simplifications),. 

\subsection{Description of the main algorithm}\label{sec:descr:alg}

\begin{algorithm}[p]
\caption{Main algorithm}\label{alg:main}
\begin{algorithmic}[1]
\Procedure{ApproxCover}{$P\in\XX^d_n$, $\Delta \in \RR$ }
     \State $S \gets$ \textsc{SimplifyCurve}$(P,\Delta)$
     \State $B \gets$ \textsc{GenerateCandidates}$(S,\Delta)$
     \State $k \gets 1$
     \State $\gamma \gets 110d+412$ \Comment{bound on the VC-dimension}
     \Repeat
        \State $k \gets 2k$
        \Comment{increase target size for solution}
        \State $r \gets 2k$  , $\Delta' \gets $ $\alpha \Delta$, $k' \gets  \lceil 16k\gamma\log(16k\gamma)\rceil$, 
         $i_{\max} \gets 5k\log_2(\frac{|B|}{k}) $ 
        \State $C \gets $ \textsc{kApproxCover}$(S,B,r,\Delta',k',i_{\max})$
        \Comment search solution with this size
     \Until{$C \neq \emptyset$}
     \Comment{..until we find a solution}
     \State \Return $C$
\EndProcedure
\end{algorithmic}

\begin{algorithmic}[1]
\Procedure{kApproxCover}{$S \in \XX^d_n$, $B \subset \XX^d_2$, $r, \Delta' \in \RR$, $k'$, $i_{\max} \in \NN$ }
     \State Let $\Dist_1$ be the uniform distribution over $B$ with weight function $w_1: B \rightarrow \{1\}$
     \State $i \gets 1$
     \Repeat
        \State $C \gets $ sample $k'$ elements from $\Dist_i$
        \State $t \gets $ \textsc{PointNotCovered}$( C,S,\Delta')$
        \If{$t = -1$} \Return $C$ 
         \Comment{if all points covered, return solution found}
        \EndIf
        \State $F \gets \emptyset$
                \Comment{otherwise, compute feasible set of $t$}
        \For{\textbf{each} $Q \in B$}
            \If{\textsc{IsFeasible}$(Q,S,t, \Delta')$} add $Q$ to $F$ \EndIf
        \EndFor
        \If{$\Pr[\Dist_{i}]{F} \leq \frac{1}{r}$}
            \State $\Dist_{i+1} \gets$ \textsc{WeightUpdate}$(\Dist_i, F)$
            \Comment{increase the probability of $F$}
            \State $i \gets i+1$
        \EndIf
    \Until{$i > i_{\max}$}
    \State \Return $\emptyset$
    \Comment{no solution found for this target size}
\EndProcedure
\end{algorithmic}
\end{algorithm}

\begin{algorithm}[p]
\caption{Subroutine \textsc{IsFeasible} which is called by the main algorithm}\label{alg:isfeasible}
\begin{algorithmic}[1]
\Procedure{IsFeasible}{$Q \in \XX^d_2, S \in \XX^d_n, t \in \Param{n}, \Delta' \in \RR$}
    \State $(t',i') \gets t$ \Comment{locate edge of $t$ on $S$}
    \State $J = \{  1 \leq i \leq j \leq n  \mid 1 \leq j-i \leq 4; i \geq i'-3; j \leq i'+4 \}$
    \Comment{find generating subcurves}    
    \For{$(i,j) \in J$} \Comment{check if $Q$ covers $t$ on $S$}
               \If{$t \in \Psi^{i,j}_{\Delta'}(S,Q)$} \Return true
               \EndIf
            \EndFor
            \State \Return false
\EndProcedure
\end{algorithmic}
\end{algorithm}

\addtocounter{linenumber}{-9}

\begin{table}[p]
\caption{Specification of additional subroutines used in the main algorithm}\label{tab:specs}
\begin{tabular}{l p{0.32\textwidth} p{0.35\textwidth} l}
    \hline
    Procedure & Input & Output  \\
    \hline\hline
   \textsc{SimplifyCurve}& $P \in \XX^d_{n}$, $\Delta \geq 0$  & $\Delta$-good simplification of $P$ (Def.~\ref{def:goodsimp}) \\\hline
   \textsc{GenerateCandidates} & $S \in \XX^d_{n}$, $\Delta \geq 0$ &  candidate set (Def.~\ref{def:candidates}) \\\hline
   \textsc{PointNotCovered} & $C \subset \XX^d_2$, $S \in \XX^d_{n}$, $\Delta \geq 0$ & either $t \in \Param{n} \setminus \Psi'_{\Delta}(S,C) $ or $-1$ if this set is empty (Lemma~\ref{lem:findpointstructured}) \\\hline
   \textsc{WeightUpdate} & distribution $\Dist$ given by weight function  $w: B \rightarrow \RR$, $F \subset B$ &  $\Dist'$ with $w': B \rightarrow \RR$ where weight is doubled for all elements of $F$  \\\hline
\end{tabular}
\end{table}

\label{sec:descr}

The algorithm receives as input a polygonal curve $P$ in $\RR^d$ and a parameter $\Delta \geq 0$. The goal is to compute an small set of edges $C$, such that all points on $P$ are covered by the $\Delta'$-coverage of $C$ on $P$ for some $\Delta'\in O(\Delta)$. 
The algorithm \textsc{ApproxCover} (see Algorithm~\ref{alg:main}), when called with input $P$ and $\Delta$, first computes a \mbox{$\Delta$-good} simplification $S$ of $P$ using the algorithm of Section~\ref{sec:simplifications:alg} and generates a finite subset $B$ of the candidate space $\CandidatesArg{E(S)} \subset \XX^d_2$ defined on the edges of this simplification. For this, we use the construction of the candidate set presented in Section~\ref{sec:alg_candidates}.
The algorithm then performs an exponential search with the variable $k$ that controls the target size of the solution. Starting with a constant $k$, the algorithm tries to find a solution of size approximately $k$ and if this fails, the algorithm  doubles $k$ and continues. For finding a solution with fixed target size, the algorithm \textsc{kApproxCover} is used (see Algorithm~\ref{alg:main}). This algorithm is called with the simplification $S$, the candidate set $B$ and set of parameters $r,\Delta',k',$ and $i_{\max}$.
The algorithm \textsc{kApproxCover} uses a variant of the multiplicative weight update method with a maximum number of (proper) iterations bounded by $i_{\max}$. In the $i$th iteration, we take a sample from a discrete probability distribution $\Dist_i$ that is defined on $B$ via a weight function $w_i: B \rightarrow \RR$, where the probability of an element $e \in B$ is defined as $w_i(e)/\sum_{e \in B} w_i(e)$. For the initial distribution $\Dist_1$, all weights are set to $1$, which corresponds to the uniform distribution over $B$. 
During the course of the algorithm, we will repeatedly update this distribution thereby generating distributions $\Dist_1, \Dist_2, \dots$ (up to $\Dist_{i_{\max}}$, unless the algorithm finds a solution in an earlier iteration). The update step performed by a call to subroutine \textsc{UpdateWeight} proceeds by doubling the weight of the subset $F$ of $B$. Note that this can be easily done in $O(|B|)$ time and space if we store the cumulative probability distribution explicitly in an array.

With this basic mechanism in place, the algorithm \textsc{kApproxCover} now proceeds as follows. In each iteration, the algorithm computes a set $C \subset B$ by taking $k'$ independent draws from the current distribution $\Dist_i$. Then, the algorithm checks, if $C$ is a solution to our problem by a call to the subroutine \textsc{PointNotCovered}. The subroutine should either return that all points on $S$ are in the $\Delta'$-coverage of the solution $C$, or return a point $t$ on $S$ that is not covered in this way. This can be done by computing the structured coverage $\Psi'_{\Delta'}(S,C)$ explicitly (see Lemma~\ref{lem:findpointstructured} in Section~\ref{sec:alg_coverage}). 
In the former case, the algorithm returns the solution and terminates. In the latter case, we compute the subset $F$ of candidates $B$ that would cover $t$ with respect to the subcurves that contain $t$ and which have at most $3$ edges. To compute the set $F$, we simply iterate over all elements of $B$ and check if $t$ is covered by calling the subroutine \textsc{IsFeasible} (see Algorithm \ref{alg:isfeasible} and Lemma~\ref{lem:tdist} in  Section~\ref{sec:alg_coverage}). (For technical reasons, we parametrize the curve $P$ via the edge space of the set of edges of $P$, so that we can locate the edge that contains $t$ in constant time.)
It is important that $F$ is not a multiset, so repeated additions of an element will not increase its weight.

At this point we would like to perform the weight update step which we described above with respect to the set $F$, however, we only do this if the weight of the set $F$ is  small.
If the total weight of the set $F$ is larger than a $\frac{1}{r}$-fraction of the total weight of $B$, then we simply skip the update step and continue by taking another sample from the current distribution.

\subsection{Algorithms for testing the coverage}\label{sec:alg_coverage}

We will now discuss how to compute the $\Delta$-coverage of a fixed solution and how to test for given $i,j,t$ and $\Delta$, whether $t$ is in $\Psi^{(i,j)}_\Delta$. These  algorithms are used as building blocks in our main algorithm.

For technical reasons, we need to be able to deduce the index of the edge of the curve $P$ that contains a point $P(t)$ in constant time from the parameter $t \in [0,1]$. To this end, we introduce the following more structured edge-parametrization.

\begin{definition}
Let $P:[0,1]\rightarrow \RR^d$. Let $(t_1,\ldots,t_n)$ be the vertex-parameters of $P$. Define the \emph{edge-parametrization} $\eta_P:\Param{n}\rightarrow [0,1]$ via $\eta(i,t)=(1-t)t_i + tt_{i+1}$. This induces a function $P\circ\eta_P:\Param{n}\rightarrow [0,1]$.
\end{definition}

The cell $C_{i,j}$ in the $\Delta$-free space of $P$ and $Q$ corresponding to the $i$th edge of $P$ and $j$th edge of $Q$ is then defined as $C_{i,j} = [\eta_P(i,0),\eta_P(i,1)]\times[\eta_Q(j,0),\eta_Q(j,1)]$.

\begin{lemma}\label{lem:tdist}
    Let $P$ be a polygonal curve in $\RR^d$, and let $(t_1,\ldots,t_n)$ be the vertex-parameters of $P$. Let $Q$ be a point in $\CandidatesArg{E}$ for some edge set $E$ and let $\Delta \geq 0$ be a real value. 
    Assume that $P$ is given as a pointer to an array storing the sequence of vertices.    
    Given any integer values $1 \leq i < j \leq n$, real value $t \in [0,1]$, there exists an algorithm that decides if $t\in \Psi^{(i,j)}_\Delta(P,Q)$, in  $O(|j-i|)$ time.
\end{lemma}

\begin{proof}
Note that the $\Delta$-free space diagram of $P$ and the edge $Q$ consists of a single row of free space cells and each free space cell can be computed locally from the two corresponding edges. 
We need to check if there exists a point $(b,1)$ on the top boundary with $b \in [\max(t,t_{j-1}),t_j]$ that is reachable by a bi-monotone path in the $\Delta$-free space which starts in a point $(a,0)$ with $a \in [t_i,\min(t,t_{i+1})]$ on the bottom boundary of the free space diagram. Using the technique by Alt and Godau~\cite{AltG95} we can process the subset of the free space diagram that corresponds to $P[t_i,t_j]$ from left to right to check if there exists such a path. If there exists such a path, then we return ``yes'', otherwise we return ``no''.
\end{proof}

\begin{lemma}\label{lem:findpointstructured}
    Given a polygonal curve $P \in \XX_n^d$, a set $C \subset \XX^d_2$ with $|C|=k$ and a real value $\Delta\geq 0$, there exists an algorithm that computes the structured $\Delta$-coverage  $\Psi'_{\Delta}(P,C)$
    in $O(n k \log( k))$ time and $O(n k)$ space.
\end{lemma}

\begin{proof}
Let $J = \{  1 \leq i \leq j \leq n  \mid 0 \leq j-i \leq 3 \}$. Fix an element $q\in C$ and an element $(i,j)\in J$. We have to compute \[R(q,i,j)=\bigcup_{a,b \in [0,1]} \{ s \in [t,t'] \mid  t=\eta_P(a,i)), t'=\eta_P(b,j), d_F(\curveP[t,t'], q) \leq \Delta,  t \leq t' \}.\] 
Consider the free space intervals $I^{h}_{i,0}=[a_{i,0},b_{i,0}]$,  $I^h_{j,1}=[a_{j,1},b_{j,1}]$,  $I^{v}_{i,1}=[c_{i,1},d_{i,1}]$ and $I^{v}_{j-1,1}=[c_{j-1,1},d_{j-1,1}]$ of the $\Delta$-free space $\mathcal{D}_{\Delta}(P,q)$. If any of the Intervals is empty or $c_{i,1}\leq d_{j-1,1}$, then we have $R(q,i,j)=\emptyset$, since there is no bi-monotone path in the free space from $I^{h}_{i,0}$ to  $I^h_{j,1}$. Otherwise, we have $R(q,i,j)=[a_{i,0},b_{j,1}]$. Each free space interval corresponds to the intersection of a line with a ball and can be computed in constant time. So in total, also $R(q,i,j)$ is an interval that can be computed in $O(1)$ time. All $R(q,i,j)$ can therefore be computed in $O(|C||J|)$ time and need $O(|C||J|)$ space.

Given  $R(q,i,j)$ for all $q\in C$ and $(i,j)\in J$, we can then compute the structured $\Delta$-coverage  $\Psi'_{\Delta}(P,C)$ with a standard scan algorithm over the computed intervals in $O(|C||J|\log(|C|))$ time. For the derivation of this bound on the running time, note that the number of overlapping intervals at any point of the scan is bounded by $O(|C|)$, since any point on the $i'$-th edge of $P$ can only be covered by an interval $R(q,i,j)$ with $i\leq i'\leq j$. The theorem statement follows by $|C|=k$ and $|J|=O(n)$.
\end{proof}

\subsection{Simplification algorithm}\label{sec:simplifications:alg}

In this section we describe an algorithm to construct a $\Delta$-good simplification $S$ for a given polygonal curve $P$.
Our simplification algorithm utilizes a data structure that is built on the input curve and which allows to query the Fr\'echet distance of a subcurve to an edge (up to some small approximation factor). For this we use the following result by Driemel and Har-Peled~\cite{jaywalking}.

\begin{theorem}[Theorem 5.9 in \cite{jaywalking}]\label{thm:DriemelSFDapprox}
    Given a polygonal curve $Z$ with $n$ vertices in $\RR^d$, one can build a data structure, in $O(\chi^2n\log^2 n)$ time, that uses $O(n\chi^2)$ space, such that for a query edge $\overline{p\,q}$, and any two points $u$ and $v$ on the curve, one can $(1+\eps)$-approximate the distance $d_F(Z[u,v],\overline{p\,q})$ in $O(\eps^{-2}\log n \log\log n)$ time, and $\chi=\eps^{-d}\log(\eps^{-1})$.
\end{theorem}

\begin{algorithm}
\caption{Curve simplification}
\label{alg:simplify}
\begin{algorithmic}[1]
\Procedure{SimplifyCurve}{Polygonal curve $P$ in $\RR^d$,$\Delta>0$}
    \State Build data structure $\mathcal{D}$ of Theorem \ref{thm:DriemelSFDapprox} on $P$ with $\eps=1/3$.
    \State Let $\hat{S}$ be an empty stack. 
    \State $\hat{S}$.\textsc{push}($1$)
    \For{$2\leq i\leq n$}\label{line:forloop}
        \State $j \gets \hat{S}$.\textsc{next\_to\_top}() 
        \While{$j$ is defined and $\mathcal{D}$.\textsc{query}$(P[t_j,t_i],\overline{p_j\,p_i}) \leq (8/3)\Delta$}\label{line:whileloop} 
            \State $\hat{S}$.\textsc{pop}()\label{line:removei}
            \State $j\gets \hat{S}$.\textsc{next\_to\_top}()
        \EndWhile
        \State $j \gets \hat{S}$.\textsc{top}()
        \If{$\|P(t_j)-P(t_i)\|\geq\Delta/3$}
            \State $\hat{S}$.\textsc{push}($i$)\label{line:filtershort}
        \EndIf
    \EndFor
    \State \Return the simplification $S$ defined by the indices in $\hat{S}$
\EndProcedure
\end{algorithmic}
\end{algorithm}
        
\begin{restatable}{theorem}{algsimpl}
\label{thm:runt:simpl}
    Let $P$ be a polygonal curve in $\RR^d$. Let $(t_1,\ldots,t_n)$ be the vertex-parameters of $P$ and $p_i=P(t_i)$ its vertices. Let $\Delta>0$ be given. There exists an algorithm that outputs an index set defining a $\Delta$-good simplification of $P$. Furthermore it does so in $O(n\log^2n)$ time and $O(n)$ space.
\end{restatable}

\begin{proof}
    Consider Algorithm \ref{alg:simplify}.
    We want to show, that the simplification $S$ of $P$ definied by the index set $\hat{S}$ is $\Delta$-good. For this we have to show, that $S$ fulfils properties $(i)-(iv)$  of Definition \ref{def:goodsimp}.
    
    Denote by $s$ the last item of $\hat{S}$, which is updated whenever $\hat{S}$ changes.
    
    Note that property $(i)$ follows immediately, as otherwise, the index $i_{j+1}$ would not have been added to $I$ in line \ref{line:filtershort}.
    
    For property $(ii)$ we will show the following invariance. 
    Whenever we start some generic iteration of the loop in line \ref{line:forloop}, where we try to add $i$ to the index set, then $P[t_s,t_{i-1}]\subset\disk[3\Delta]{p_s}$. At the start of the first iteration, $\hat{S}=(1)$ and $i=2$. As $P[t_1,t_{i-1}] = P[t_1,t_1]=p_1$, the invariance holds.
    
    Now assume we are at the start of some iteration, where we try to add $i$ to $\hat{S}$, and assume the invariance holds.
    After exiting the loop in line \ref{line:whileloop}, we either updated $\hat{S}$, or we did not. In any case, we will assume, that $\|p_s-p_i\|<\Delta/3$, because otherwise, we add $i$ to $\hat{S}$ in this iteration. But then at the start of the next iteration, in which we consider adding $i+1$, we then have that $s=i$, for which the invariance trivially holds, as $P[t_s,t_{(i+1)-1}] = P[t_s,t_s] = p_s$.
    
    Assume for now, that we did not update $\hat{S}$, that is we never entered line \ref{line:removei}. Then $P[t_s,t_{i}]\subset\disk[3\Delta]{p_s}$, as $P[t_s,t_{i-1}]$, $p_{i-1}$ and $p_i$ lie inside $\disk[3\Delta]{p_s}$.
    
    Otherwise, $d_F(P[t_s,t_{i}],\overline{p_s\,p_{i}})\leq \textsc{query}(P[t_s,t_{i}],\overline{p_s\,p_{i}})\leq 8\Delta/3$, implying together with our assumption $\|p_s-p_i\|<\Delta/3$, that $P[t_s,t_{i}]\subset\disk[3\Delta]{p_s}$, hence the invariance holds in every iteration.
    
    This invariance implies, that whenever we start iteration $i$ of the loop, we have that 
    \begin{align*}
        d_F(P[t_s,t_i],\overline{p_s\,p_i})&\leq \max\left(d_F(P[t_s,t_{i-1}],\overline{p_s,p_s}) , d_F(P[t_{i-1},t_{i}],\overline{p_s,p_{i}}) \right)\\&\leq \max\left(3\Delta , d_F(P[t_{i-1},t_{i}],\overline{p_s,p_{i}}) \right)\\
        &\leq \max\left(3\Delta ,\Delta/3\right)=3\Delta,
    \end{align*}
    where the second inequality follows form the invariance, as $\overline{p_s\,p_s} =p_s$, and the third from the fact, that $P[t_{i-1},t_i]$ is an edge from $p_{i-1}$ to $p_{i}$, and we know, that $\|p_s-p_{i-1}\|\leq\Delta/3$.
    Thus at the start of each iteration $i$ we have that 
    \[ 
    d_F(P[t_s,t_i],\overline{p_s\,p_i})\leq3\Delta\label{eq:ineq}\tag{*}
    \]
    holds for $i$ and $s$. As we continuously remove the last item from $\hat{S}$, we only do so, if \eqref{eq:ineq} holds for $\hat{S}$.next\_to\_top() and $i$. As we remove $s$, $\hat{S}$.next\_to\_top() takes the place of $s$, thus at every iteration of the loop in line \ref{line:whileloop}, as well as when we exit the loop, \eqref{eq:ineq} holds. Thus every time we possibly add any index in line \ref{line:filtershort}, \eqref{eq:ineq} holds, and thus property $(ii)$ is always maintained.
    
    Property $(iii)$ follows directly for $i_1=1$. As $i_k$ may be less than $n$, the property follows from the invariance, as $\overline{p_{i_k}\,p_{i_k}}=p_{i_k}$. And thus $P[t_{i_k},t_n]\subset\disk[3\Delta]{p_{i_k}}$.
    
    For property $(iv)$ observe that, if $d_F(P[t_{i_j},t_{i_{j+2}}],\overline{p_{i_j}\,p_{i_{j+2}}}) < 2\Delta$ for some $i_j$ and $i_{j+2}$ in the resulting index set $I$, the algorithm would have removed $i_{j+1}$ from $I$ in line \ref{line:whileloop}, as when we add $i_{j+2}$, $\hat{S}.\mathrm{next\_to\_top}()=i_{j}$, and hence
    \[\textsc{query}(P[t_{i_j},t_{i_{j+2}}],\overline{p_{i_j}\,p_{i_{j+2}}}) \leq(4/3)d_F(P[t_{i_j},t_{i_{j+2}}],\overline{p_{i_j}\,p_{i_{j+2}}}) < 8/3\Delta.\]
    
    The space is dominated by the space for storing the data structure. For analysing the running time, note that each vertex of $P$ is inserted to and removed from the index set at most once. Therefore, the total running time is bounded by the preprocessing time and the at most $O(n)$ queries to the data structure. Thus the iterations of the loop in line \ref{line:whileloop} are bounded by $O(n)$ overall. Thus the claim follows from Theorem \ref{thm:DriemelSFDapprox}.
\end{proof}

\section{Analysis of the main algorithm}\label{sec:analysis:main}

The algorithm described in Section~\ref{sec:descr:alg} is based on the set cover algorithm by Brönniman and Goodrich~\cite{bronnimann1995almost}. In Section~\ref{sec:overview:ana:main}, we first give an overview of the main ideas of the analysis. A crucial step  is the analysis of the VC-dimension of the dual set system. In our case, this is a set system that is formed by the sets $F$ computed in the main algorithm. We will perform this analysis in Sections~\ref{sec:struc:feasible} and \ref{sec:vcdim}. With proper bounds on the VC-dimension in place, we then proceed with the analysis of the main algorithm in Section~\ref{sec:analysis:general}. In Section~\ref{sec:analysis:cpacked} we give improved bounds which are  tailored to $c$-packed polygonal curves.

\subsection{Overview of the proof of the main theorem}\label{sec:overview:ana:main}

For the formal analysis of the set system that is formed by the sets $F$ computed in Algorithm~\ref{alg:main}, we introduce the notion of feasible sets.

\begin{definition}[Feasible set]
Let $S: \Param{n} \rightarrow \RR^d$ be a polygonal curve and  let $B \subset \XX^d_2$ be a candidate set of edges and let $\Delta \geq 0$ be a real value. For any point $t\in \Param{n}$, we define the \emph{feasible set} of $t$ as the set of elements $Q\in B$ that admit an $(i,j)$-partial traversal with $S$ that fully covers $Q$ and that covers $t$ on $S$, with the additional condition that $j-i \leq 3$. We denote the feasible set of $t$ with $F_{\Delta}(t)$. 
\end{definition}

Note that for any fixed $S$ and $\Delta$ the set of feasible sets
$ \{F_{\Delta} (t) \mid t \in \Param{n} \} $
is exactly the set system determined by the subroutine \textsc{IsFeasible} described in Algorithm~\ref{alg:isfeasible}.

We claim that any feasible set can be split into sets corresponding to the edges of the simplification, where each set consists of a constant union of rectangles in the candidate space restricted to the respective edge. Figure~\ref{fig:freespace:rect}  illustrates one of those rectangles. The following lemma provides the formal statement.

\begin{restatable}{lemma}{rectangles}
\label{lem:rect}
Let $P$ be a polygonal curve in $\RR^d$ and let $e\in\XX^d_2$  be an edge. Let $(t_1,\ldots,t_n)$ be the vertex-parameters of $P$.
For any  integer values $1 \leq i < j \leq \min(i+3, n)$ and real value $t \in [0,1]$ with  $t_i \leq t \leq t_j$,
either there exist $\alpha_1,\alpha_2,\beta_1,\beta_2$ such that 
\[ R:= \{(\alpha,\beta) \in [0,1]^2 \mid t\in\Psi^{i,j}_\Delta\left( P,e[\alpha,\beta] \right)\} = [\alpha_1,\alpha_2] \times [\beta_1,\beta_2], \]
or the set $R$ is empty.
Moreover, each $\alpha_v$ (respectively $\beta_v$) for $v\in\{1,2\}$ can be written as $\alpha_v=c_v+\sqrt{d_v}$ (respectively $\beta_v=e_v+\sqrt{f_v}$), where the parameters $c_v$ and $d_v$ (respectively $e_v$ and $f_v$) can be computed by an algorithm that takes $(i,j),t$ and $e$ as input and needs $O(d)$ simple operations.
\end{restatable}

\begin{figure}[t]
    \centering
    \includegraphics[width=1\textwidth]{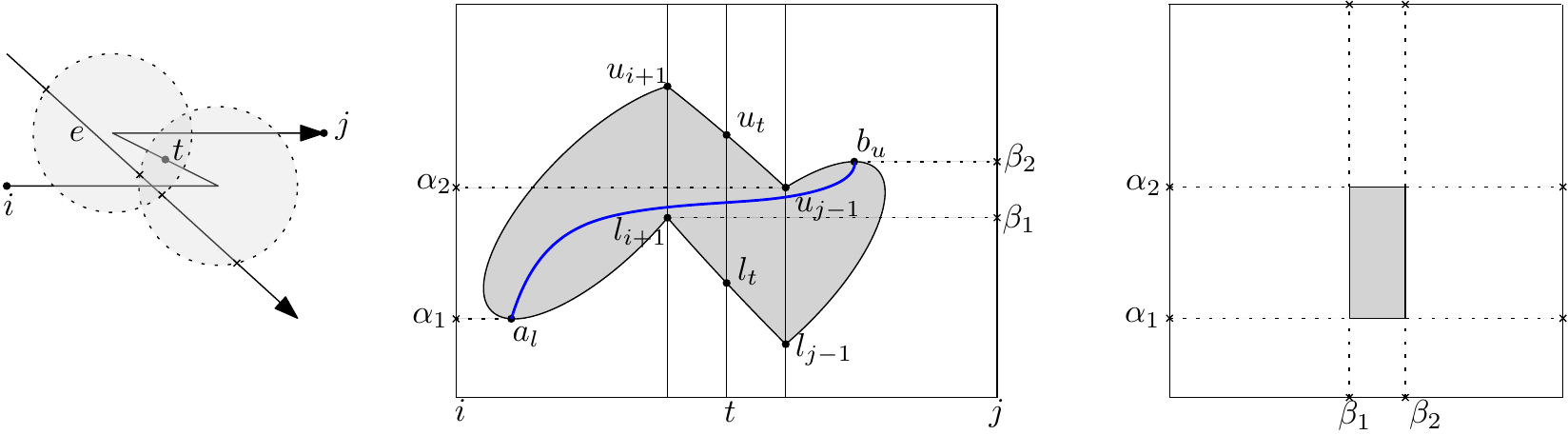}
    \caption{Example for the construction of the rectangle $R=[\alpha_1,\alpha_2]\times[\beta_1,\beta_2]$ for fixed $P,i,j,t,\Delta$ and $e$. The left image shows the curves $P[t_i,t_j]$ and $e$ with two circles of radius $\Delta$ around $P(t_{i+1})$ and $P(t_{j-1})$. The middle image shows the corresponding $\Delta$-free space diagram with a $(i,j)$-partial traversal from $a_l$ to $b_u$ and the right image shows the rectangle $R$ in the parameter space  $[0,1]^2$ of $e$.}
    \label{fig:freespace:rect}
\end{figure}

To prove a VC-dimension bound of $O(d)$, we combine the above lemma with the following general theorem which can be attributed to Goldberg and Jerrum~\cite{GJ95}. We use the variant by Anthony and Bartlett~\cite{AB99}, which is stated as follows.

\begin{restatable}{theorem}{anthonybartlett}[Theorem 8.4 \cite{AB99}]\label{thm:vcsimpl}
Suppose $h$ is a function from $\RR^a\times \RR^b$ to $\{0,1\}$ and let $H=\{x\rightarrow h(\alpha,x)\;|\;\alpha\in\RR^a\}$
be the class determined by $h$. Suppose that $h$ can be computed by an algorithm that takes as input the pair $(\alpha,x)\in \RR^a\times\RR^b$ and returns $h(\alpha,x)$ after no more than $t$ simple operations. Then, the VC-dimension of $H$ is $\leq 4a(t+2)$.
\end{restatable}

Our bound on the VC-dimension can be found in Lemma~\ref{lem:vcdimfeas} with a more detailed analysis of the constant factors in Lemma~\ref{lem:vcdim:const}. Using these bounds on the VC-dimension, we then obtain the following main result. The proof can be found in Section~\ref{sec:analysis:general}.

\begin{restatable}{theorem}{mainalg}\label{thm:runtime:main:canon}
Given a polygonal curve $P\in \XX^d_n$ and  $\Delta\in\RR_+$, there exists an algorithm that computes an 
$(\alpha, \beta)$-approximate solution to the $\Delta$-coverage problem on $P$ with $\alpha=11$ and $\beta = O(\log k^*)$, 
where $k^{*}$ is the minimum size of a solution to the $\Delta$-coverage problem on $P$.
The algorithm needs in expectation  $\widetilde{O}(({k^*})^2 n +k^* n^3)$ time and $ \widetilde{O}(({k^*}) n +n^3)$ space.
\end{restatable}

The proof of  Theorem~\ref{thm:runtime:main:canon} uses the well-known $\frac{1}{r}$-net theorem by Haussler and Welzl~\cite{haussler1987eps} (see Theorem~\ref{thm:epsnet}, below), which provides a bound on the probability that our sample chosen in line 5 is a $\frac{1}{r}$-net of the weighted set system, based on the VC-dimension of this set system. We use this to bound the expected number of iterations of the main loop in \textsc{kApproxCover} within our analysis of the multiplicative weights update algorithm.

Theorem~\ref{thm:epsnet} has been extended and improved, in particular by Li, Long, and Srinivasan~\cite{LI2001516}, see also the survey by Mustafa and Varadarajan~\cite{mustafa2017epsilonapproximations,toth2017handbook}. However, it seems that Theorem~\ref{thm:epsnet} gives the best bound for our purposes.

\begin{theorem}[\cite{haussler1987eps}]\label{thm:epsnet}
For any $(X,\RSpace)$ of finite VC-dimension $\delta$, finite $B \subseteq X$ and $0<\frac{1}{r}$, $\alpha<1$, if $N$ is a subset of $B$ obtained by at least
\[   \max\left(\frac{4}{3} \log\left(\frac{2}{\alpha}\right),8\delta r\log\left(8\delta r\right)\right) \]
random independent draws, then $N$ is an $\frac{1}{r}$-net of $B$ for $\RSpace$ with probability at least $1-\alpha$.
\end{theorem}

\subsection{On the structure of feasible sets}\label{sec:struc:feasible}

To analyse the VC-dimension of the feasible sets, we first show Lemma~\ref{lem:rect}.

\begin{proof}[Proof of Lemma~\ref{lem:rect}]
    Let $R=\{(\alpha,\beta) \in [0,1]^2 \mid t\in\Psi^{i,j}_{\Delta}\left( P,e[\alpha,\beta] \right) \}$.
     We first show that either $R$ corresponds to a rectangle $[\alpha_1,\alpha_2]\times[\beta_1,\beta_2]$ in the parameter space $[0,1]^2$ of $e$ or $R=\emptyset$. In Figure~\ref{fig:freespace:rect}, we give an example for the construction of the rectangle $R$. Let $i<v<j$. The intersection of the $\Delta$-free space $\dfree[\Delta]{}{}(P,e)$ with the edge $t_v\times[0,1]$ is either a free space interval of the form $t_v\times[l_v,u_v]$ or is empty. Also the intersection with $t\times[0,1]$ has the form $t\times[l_t,u_t]$ or is empty. If either of the intersections is empty for $t$ or some $i< v< j$ then  $R$ is empty, since no point on $e$ is within distance $\Delta$ of $P(t)$ respectively $P(t_v)$. Otherwise for all $i< v < j$ the parameters $l_v$,$u_v$, as well as $l_t$ and $u_t$ are well defined. In the case that $j-i=3$  and $l_{i+1}>u_{i+2}$ we have $R$ is empty since there is no bi-monotone path in the free space that first passes $t_{i+1}\times[l_{i+1},u_{i+1}]$ and then $t_{i+2}\times[l_{i+2},u_{i+2}]$. In the following we therefore assume $l_{i+1}\leq u_{i+2}$ for $j-i=3$. Let further $a=(a_i,a_l)$ be the lowest point in the cell of the $\Delta$-free space corresponding to the edge $P[t_i,t_{i+1}]$ and $b=(b_j,b_u)$ be the highest point in the cell of the $\Delta$-free space corresponding to the edge $P[t_{j-1},t_j]$.  We define the following parameters
    \begin{align*}
        \alpha_1 &=\begin{cases}l_t &\text{for}\  a_i\geq t\\a_l & \text{else} \end{cases}\\
        \alpha_2 &=\min(u_{i+1},\dots,u_{j-1},u_t) &\\
        \beta_1 &=\max(l_{i+1},\dots,l_{j-1},l_t) &\\
        \beta_2 &=\begin{cases}u_t &\text{for}\ b_j\leq t\\b_u & \text{else} \end{cases}
    \end{align*}
    and show that $R$ corresponds to a rectangle $[\alpha_1,\alpha_2]\times[\beta_1,\beta_2]$ in the parameter space $[0,1]^2$ of $e$.
    To do so we first show that for each $\alpha\in[\alpha_1,\alpha_2]$ and $\beta\in[\beta_1,\beta_2]$ with $\alpha\leq \beta$ there is a $t_\alpha\in[t_i,t_{i+1}]$ and a $t_\beta\in[t_{j-1},t_j]$ such that \[d_F(P[t_\alpha,t_\beta],e)\leq \Delta.\]
    The case $\beta< \alpha$ is analogous. So let $\alpha\in[\alpha_1,\alpha_2]$, $\beta\in[\beta_1,\beta_2]$ with $\alpha\leq\beta$. By the definition of $\alpha_1,\alpha_2$ it directly follows that there is a $t_\alpha\in[t_i,\min(t_{i+1},t)]$ such that $(t_\alpha,\alpha)$ is in the free space $\dfree[\Delta]{}{}(P,e)$. By the definition of $\beta_1,\beta_2$ it also follows that there is a $t_\beta\in[\max(t_{j-1},t),t_j]$ such that $(t_\beta,\beta)$ is in the free space $\dfree[\Delta]{}{}(P,e)$. We split the analysis on how to construct a monotone increasing path in the free space from $(t_\alpha,\alpha)$ to $(t_\beta,\beta)$ in three cases depending on $j-i$.
    
    Case $j-i=1$: Since the freespace is convex in each cell and $\alpha\leq\beta$, the path  $(t_\alpha,\alpha)\oplus(t_\beta,\beta)$ is a bi-monotone path in $\dfree[\Delta]{}{}(P,e)$.

    Case $j-i=2$: Let $\gamma=\min(u_{i+1},\beta)$. Since $\beta\geq \beta_1\geq l_{i+1}$ and therefore $u_{i+1}\geq \gamma \geq l_{i+1}$, we have $(t_{i+1},\gamma)\in \dfree[\Delta]{}{}(P,e)$. By the convexity of each cell in the free space and the fact that $\alpha\leq \gamma\leq \beta$, we get that \[(t_\alpha,\alpha)\oplus(t_{i+1},\gamma)\oplus(t_\beta,\beta)\] is a monotone increasing path in $\dfree[\Delta]{}{}(P,e)$.
    
    Case $j-i=3$: Let $\gamma_1=\min(u_{i+1},u_{i+2},\beta)$ and $\gamma_2=\min(u_{i+2},\beta)$. By $l_{i+1}\leq \beta_{i+1}\leq \beta$ and $l_{i+1}\leq u_{i+2}$, we get $(t_{i+1},\gamma_1)\in\dfree[\Delta]{}{}(P,e)$. Further by $l_{i+2}\leq \beta_1\leq \beta$, we get $(t_{i+2},\gamma_2)\in\dfree[\Delta]{}{}(P,e)$. By the convexity of each cell in the free space diagram and the fact $\alpha\leq \gamma_1\leq\gamma_2\leq \beta$, we get similar to the last case that
    \[(t_\alpha,\alpha)\oplus(t_{i+1},\gamma_1)\oplus(t_{i+2},\gamma_2)\oplus(t_\beta,\beta)\] is a bi-monotone  path in $\dfree[\Delta]{}{}(P,e)$.
    
    It remains to show that no other $(\alpha,\beta)\in[0,1]^2$ with $\alpha\leq\beta$ and $e[\alpha,\beta]\in R$ exist. So we show that for $\alpha\notin[\alpha_1,\alpha_2]$ there is no $t_\alpha \in [t_i,t_{i+1}]$, $t_\beta \in [t_{j-1},t_j]$, $\beta\geq \alpha$ such that \[d_F(P[t_\alpha,t_\beta],e[\alpha,\beta])\leq \Delta.\]
    
    Case $\alpha<\alpha_1$: By the definition of $\alpha_1$ there exists no  $t_\alpha\in[t_i,t_{i+1}]$ such that $\|P(t_\alpha)-e(\alpha)\|\leq \Delta$. 
    
    Case $\alpha>\alpha_2$: Assume there exists a $i<v<j$ such that $\alpha_2=u_v$. Then by the definition of $u_v$ there exists no $\gamma\in e[\alpha,1]$ such that $\|P(t_v)-e(\gamma)\|\leq \Delta'$. Otherwise we have $\alpha_2=u_t$. Then by the definition of $u_t$, there exists no $\gamma\in e[\alpha,1]$ such that $\|P(t_v)-e(\gamma)\|\leq \Delta$. The combination of the two cases directly implies the claim.
    
    Now we analyse the running time to compute the parameters $c_1,c_2,d_1,d_2,e_1,e_2,f_1,f_2$ that define $\alpha_1,\alpha_2,\beta_1,\beta_2$. Each of the parameters $\alpha_1,\alpha_2,\beta_1,\beta_2$ is equal to one parameter of the set  $K=\{u_{i+1},\dots,u_{j-1},l_{i+1},\dots,l_{j-1},a_l,b_u,l_t,u_t\}$ with $|K|\leq 8$. Each of the intervals $[l_v,u_v]$ as well as $[l_t,u_t]$ correspond to the intersections of a line with a ball in $\RR^d$. The points $a=(a_i,a_l)$ and $b=(b_i,b_j)$ are defined by the intersections of a line with the union of a cylinder and two balls in $\RR^d$. Each extreme point $\kappa\in K$ of such an intersection can be written as $\kappa=\kappa_1+\sqrt{\kappa_2}$ with parameters $\kappa_1$ and $\kappa_2$ that can be computed with $O(d)$ simple operations (see Lemma 17 and 18 in \cite{driemel2019vc}). For each parameter $\gamma\in \{\alpha_1,\alpha_2,\beta_1,\beta_2\}$, we can find $\kappa^*\in K$ with $\kappa^*=\gamma$ by comparisons of $O(1)$ elements of $K$ with each other. For each two element of $\kappa,\kappa'\in K$ this can be done with $O(1)$ simple operations, given the parameters $\kappa_1,\kappa_2,\kappa_1',\kappa_2'$ with $\kappa=\kappa_1+\sqrt{\kappa_2}$ and $\kappa'=\kappa_1'+\sqrt{\kappa_2'}$. So in total we need $O(d)$ simple operations to compute the parameters $c_1,c_2,d_1,d_2,e_1,e_2,f_1,f_2$ that define $\alpha_1,\alpha_2,\beta_1,\beta_2$.
\end{proof}

\subsection{Analysis of the VC-dimension}\label{sec:vcdim}
We analyse the VC-dimension of the set system of feasible sets $ \{F_{\Delta} (t) \mid t \in \Param{n} \} $.
This VC-dimension can be used to bound the probability that in any iteration of the algorithm $F$ is larger than a $\frac{1}{r}$-fraction of the total weight of $B$ and the update step would get skipped. The following general theorem can be attributed to Goldberg and Jerrum~\cite{GJ95}. We use the variant by Anthony and Bartlett stated in Theorem~\ref{thm:vcsimpl}.

\begin{lemma}\label{lem:vcdimfeas}\label{cor:vccand}
Let $S:\Param{n} \rightarrow \RR^d$ be a polygonal curve and let $\Delta \in \RR_{+}$.
Consider the set system $ \{F_{\Delta} (t) \mid t \in \Param{n} \} $ with ground set $\XX^d_2$.
The VC-dimension of this set system is in $O(d)$.
\end{lemma}
\begin{proof}
Define a function
$h: \Param{n} \times \XX^d_2 \rightarrow \{0,1\}$ with
$h(t,Q)=1$  if and only if a call to \textsc{IsFeasible}$(Q,S,t,\Delta)$ returns true.
We analyse the VC-dimension of the class of functions determined by $h$:
\[H =\{x \rightarrow h(t, x) \;|\; t \in \Param{n}\}\]
As a consequence, we obtain the same bounds  on the VC-dimension of the corresponding set system $\mathcal{R}$ with ground set $\XX^d_2$ where a set $r_t \in \mathcal{R}$ is defined by a $t \in \Param{n}$ with \[ r_{t} = \{ Q \in \XX^d_2 \mid h(t,Q)=1\} \] 

In order to show the lemma, we first argue that for any given $t \in \Param{n}$ and $Q \in \XX^2_d$ the expression $h(t,Q)$ can be evaluated with $O(d)$ simple operations.

Let $(t',i') = t$ and recall the index set $J = \{ (i,j) \mid i'-3 \leq i \leq i' \leq j \leq i+3 \}$ as in the procedure \textsc{IsFeasible}. Note that $|J|=9$ and that $J$ can be determined by $O(1)$ simple operations from $i'$.
Note that \textsc{IsFeasible} returns true if and only if $t\in\Psi^{i,j}_{\Delta}(S,Q)$ for some $(i,j) \in J$.
So, for fixed $(i,j)$, consider the set 
 \[ R= \{(\alpha,\beta) \in [0,1]^2 \mid t\in\Psi^{i,j}_{\Delta}\left( S,Q[\alpha,\beta] \right)\} \] Lemma~\ref{lem:rect} implies that $R$ is either empty or can be written as a rectangle $[\alpha_1,\alpha_2]\times [\beta_1,\beta_2]$. Note that $t\in\Psi^{i,j}_\Delta(S,Q)$ if and only if $R$ is non-empty and $(0,1) \in R$. By Lemma~\ref{lem:rect}, this test can be performed using $O(d)$ simple operations.
 Thus, we can apply Theorem~\ref{thm:vcsimpl} and conclude that the VC-dimension of  $H$ is in $O(d)$.
\end{proof}

\subsection{Detailed analysis of the VC-dimension}\label{sec:feas}

In Lemma~\ref{lem:vcdimfeas}, it was shown that the VC-dimension of the set induced by the function \textsc{IsFeasible} is in $O(d)$. In this section, we explicitly derive constants $a$ and $b$ such that this VC-dimension is at most $ad+b$. This is necessary to obtain exact bounds on the sample size to be used in the main algorithm. To do so, we give a more detailed pseudocode of the procedure \textsc{IsFeasible} in Algorithm~\ref{alg:feasible} and analyse the simple operations that are need for each step of the algorithm.
In the algorithm, we refer to specific free space intervals by using the following notation. Let $C_{i}$ be the cell in the $\Delta$-free space of a curve $S$ and an edge $e$ corresponding to the $i$th edge of $S$ and $e$. We denote with $I^{h}_{i,0}=[a_{i,0},b_{i,0}]$ the horizontal free space interval that bounds $C_{i}$ from below and with $I^{h}_{i,1}=[a_{i,1},b_{i,1}]$ the horizontal free space interval that bounds $C_{i}$ from above. We further denote with $I^{v}_{i,0}=[c_{i,0},d_{i,0}]$ the vertical free space interval that bounds $C_{i}$ from the left and with $I^{v}_{i,1}=[c_{i,1},d_{i,1}]$ the vertical free space interval that bounds $C_{i}$ from the right.

The following lemmata give us bounds on the number of iterations for specific operations of the algorithm. Figure~\ref{fig:freespace:feasset} shows an example of the relevant free space intervals that are used by the algorithm for checking if $t$ is contained in the structured coverage of some $(i,j)\in \{ (i,j) \mid i'-3 \leq i \leq i' \leq j \leq i+3       \}$.

\begin{algorithm}[h]
\caption{Alternative definition of feasibility test}\label{alg:feasible}
\begin{algorithmic}[1]
\Procedure{IsFeasible}{$e \in \XX^d_2, S \in \XX^d_2, t \in \Param{n}, \Delta \in \RR$}
    \State Denote with $I^{h}_{i,0}, I^h_{j,1}, I^{v}_{i,1}, I^{v}_{j-1,1}$ free space intervals of $\mathcal{D}_{\Delta}(e,S)$ \label{line:comp:int}
    \State $(t',i') \gets t$
    \State $J = \{ (i,j) \mid i'-3 \leq i \leq i' \leq j \leq i+3       \}$
    \For{$(i,j) \in J$}
        \If{$I^{h}_{i,0} \neq \emptyset$ and  $I^h_{j,1} \neq \emptyset$} \label{line:checkempty1}
        \If{$a_{i,0} \leq t$ and $t \leq b_{j,1}$} \label{line:check:t}
                \If{$i=j$} \label{line:check:ij}
                    \State \Return true 
                \EndIf
                \If{$I^{v}_{i,1} \neq \emptyset$} \label{line:checkempty2}
                    \If{$j=i+1$} \label{line:check:ji+1}
                    \State \Return true
                    \EndIf
                    \If{$I^{v}_{j-1,1}\neq \emptyset$} \label{line:checkempty3}
                    \If{$c_{i,1} \leq d_{j-1,1}$  } \label{line:check:cd}
                    \State \Return true
                    \EndIf
                    \EndIf
                \EndIf
            
        \EndIf
        \EndIf
    \EndFor
    \State \Return false
\EndProcedure
\end{algorithmic}
\end{algorithm}

\begin{lemma}\label{lem:abc}
Let $a,b,c\in\RR$ with $b\geq 0$. The truth value of $a+\sqrt{b}\leq c$ can be computed by an algorithm that takes $a,b$ and $c$ as input and needs at most $4$ simple operations.
\end{lemma}
\begin{proof}
    To check if $a+\sqrt{b}\leq c$ we can first check if $a>c$ with one simple operations. If that is the case we can return false. If it is not the case, the statement is equivalent to $b\leq (c-a)^2$. We need one subtraction and one multiplication to calculate $(c-a)^2$. Together with the comparison we get a total of $4$ simple operations.
\end{proof}

\begin{lemma}\label{lem:abcd}
Let $a,b,c,d\in\RR$ with $b,d\geq 0$. The truth value of $a+\sqrt{b}\leq c+\sqrt{d}$ can be computed by an algorithm that takes $a,b,c$ and $d$ as input and needs at most $11$ simple operations.
\end{lemma}
\begin{proof}
    Consider Algorithm~\ref{alg:CheckIneq}. It is easy to check that the procedure \textsc{CheckInequality} ouputs the truth value of $a+\sqrt{b}\leq c+\sqrt{d}$. We show that the algorithm can be implemented such that it needs at most $11$ simple operations. The checks in line \ref{simp:check1}, \ref{simp:check2} and \ref{simp:check3} need one simple operation each. The computation of $z=(c-a)^2-b-d$ in line \ref{check:z1} or \ref{check:z2} can be done by 3 subtractions and one multiplication. The check for $z\geq0$ ( respectively $>$) in line \ref{check:z0:leq:1} ( respectively \ref{check:z0:leq:2}) is one simple operation. The check $z^2\leq4bd$ (respectively $<$) in the same line needs 3 multiplications and one comparison. Counting the number of simple operations together, we see that the algorithm needs at most $11$ simple operations.
\end{proof}

\begin{algorithm}[h]
\caption{Check for $a+\sqrt{b}\leq c+\sqrt{d}$}\label{alg:CheckIneq}
\begin{algorithmic}[1]
\Procedure{CheckIneqaulity}{$a,b,c,d\in \RR$}
    \If{$c\geq a$} \label{simp:check1}
        \If{$d\geq b$} \label{simp:check2}
            \State \Return true
        \EndIf
        \State Let $z=(c-a)^2-b-d$ \label{check:z1}
        \If{$z\geq 0$ \textbf{or} $z^2\leq 4bd$} \label{check:z0:leq:2}
            \State \Return true
        \Else
            \State \Return false
        \EndIf
    \Else
        \If{$d<b$}\label{simp:check3}
            \State \Return false
        \EndIf
        \State Let $z=(c-a)^2-b-d$  \label{check:z2}
        \If{$z> 0$ \textbf{or} $z^2< 4bd$} \label{check:z0:leq:1}
            \State \Return false
        \Else
            \State \Return true
        \EndIf
    \EndIf
\EndProcedure
\end{algorithmic}
\end{algorithm}

\begin{figure}[h]
    \centering
    \includegraphics{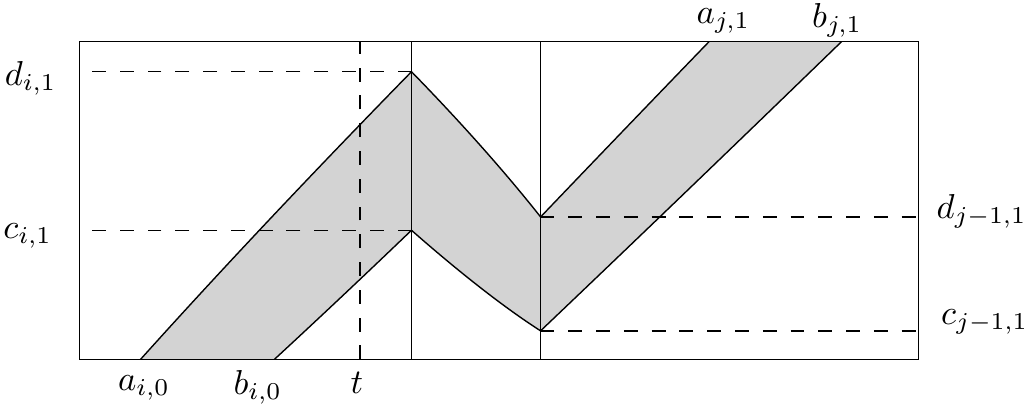}
    \caption{Example of the free space intervals $I^{h}_{i,0}=[a_{i,0},b_{i,0}]$,  $I^h_{j,1}=[a_{j,1},b_{j,1}]$,  $I^{v}_{i,1}=[c_{i,1},d_{i,1}]$ and $I^{v}_{j-1,1}=[c_{j-1,1},d_{j-1,1}]$ for some $t=(t',i')$ and some $(i,j)\in \{ (i,j) \mid i'-3 \leq i \leq i' \leq j \leq i+3       \}$}
    \label{fig:freespace:feasset}
\end{figure}

\begin{lemma}\label{lem:intersec:circ:line}
Let $p,q,r\in \RR^d$ and let $\Delta\in\RR_+$. The intersection of the line through $p$ and $q$ and the ball with radius $\Delta$ around $r$ is either 
$\{p+ t(q-p)\mid t\in [a-\sqrt{b},a+\sqrt{b}]\subset\RR\}$ for some $a,b\in\RR$ or is empty. There exists an algorithm  which needs at most $10d+5$ simple operations that takes $p,q,r$ and $\Delta$ as input and outputs the values $a$ and $b$ or decides that the  intersection is empty.
\end{lemma}
\begin{proof}
    The algorithm needs to solve the quadratic equation
    \[\|p+t(q-p)-r\|^2\leq \Delta^2\]
     which is equivalent to $t^2+\alpha t+\beta=0$ where \[\alpha=\frac{\sum_{i=1}^d 2(p_i-q_i)r_i}{\sum_{i=1}^d (q_i-p_i)} \text{ and } \beta=\frac{-\Delta^2+\sum_{i=1}^d (p_i^2+r_i^2)}{\sum_{i=1}^d (q_i-p_i)}.\]
     The sum $\sum_{i=1}^d (q_i-p_i)$ can be computed using $2d-1$ simple operations. For the sum $\sum_{i=1}^d 2(p_i-q_i)r_i$ needs $4d-1$ simple operations and the term $-\Delta^2+\sum_{i=1}^d (p_i^2+r_i^2)$ needs $4d+1$ operations. So with the two extra divisions, we need a total of $10d+1$ operations to compute  $\alpha$ and $\beta$. The solutions of the quadratic equation are
     \[t_{1,2}=-\frac{\alpha}{2}\pm \sqrt{\left(\frac{\alpha}{2}\right)^2-\beta}.\]
     So we need an extra operation to calculate $a=-\frac{\alpha}{2}$ and an extra two operations to calculate $b=a^2-\beta$. Also, we have to check if $b<0$ to decide if the solution is feasible or the intersection is empty. This sums up to a total of $10d+5$
\end{proof}

With the help of the above lemmata, we can get the following result.
\begin{lemma}\label{lem:vcdim:const}
The algorithm \textsc{IsFeasible} (Algorithm~\ref{alg:feasible}) can be implemented such that it needs at most $110d+412$ simple operations.
\end{lemma}
\begin{proof}
    In total the algorithm needs to compute at most $6$ horizontal and $5$ vertical free space intervals. Each of them corresponds to the intersection of an edge with a ball. By Lemma~\ref{lem:intersec:circ:line}, the parameters $a$ and $b$ of the intersection of the line containing this edge  with the ball can be computed with an algorithm that needs at most $10d+5$ simple operations. To get an explicit representations of the borders of the free space interval, we have to additionally check if $a+\sqrt{b}\leq 1$, $a-\sqrt{b}\geq 0$ and $a-\sqrt{b}> 1$. By Lemma~\ref{lem:abc}, each of these checks needs at most $4$ simple operations. So we need at most $10d+17$ operations to compute the parameters $c,d,e,f$ that describe a freespace interval $[c-\sqrt{d},e+\sqrt{f}]$ implicitly. In total, we therefore need $110d+187$ simple operations for all $11$ free space intervals. After computing the intervals, the 4 checks in line \ref{line:checkempty1}, \ref{line:checkempty2} and \ref{line:checkempty3} only need one simple operation each. The checks $i=j$ in line \ref{line:check:ij} and $j=i+1$ in line \ref{line:check:ji+1} are also only one simple operation each. By Lemma~\ref{lem:abc}, the two checks in line \ref{line:check:t} need $4$ simple operations each. The remaining check in line \ref{line:check:cd} needs $11$ simple operations, as shown in Lemma~\ref{lem:abcd}. So each iteration of the for loop (line 5-15) needs at most $25$ simple operations. Since $J$ has at most $9$ elements. The whole for loop needs at most $225$ simple operations. By adding the simple operations, that we need to compute the free space intervals, we get that the whole algorithm needs at most $110d+412$ simple operations.
\end{proof}

\subsection{Results for general polygonal curves}\label{sec:analysis:general}

In this section we prove Theorem~\ref{thm:runtime:main:canon} which we stated in Section~\ref{sec:overview:ana:main}. 
Instead of proving Theorem~\ref{thm:runtime:main:canon} directly, we prove the following more general lemma first.
We include this general version, since we will use it in Section~\ref{sec:implicit:weights} again to bound the running time of a modified variant of the algorithm that uses implicit weight updates.

\begin{restatable}{lemma}{mainalgor}
\label{lem:runtime:main:basic}
Consider the main algorithm \textsc{ApproxCover} with input $P\in \XX^d_n$ and $\Delta\in\RR_+$ 
and let $k^{*}$ be the minimum size of a solution to the $\Delta$-coverage problem on $P$.
Assume the generated set of candidates $B$ contains a structured $\alpha\Delta$-covering $H\subset B$ of the generated $\Delta$-good simplification $S$ with $|H|\leq \beta k$ for some constants $\alpha,\beta\in\RR$.
The algorithm outputs an 
$(\alpha', \beta')$-approximate solution with $\alpha'=\alpha+3$ and $\beta' = O(\log( k^*))$.
In addition to the time and space needed to generate the set of candidates $B$ and to generate the simplification $S$, the algorithm needs in expectation  $\widetilde{O}(({k^*})^2 n +|B|k^*)$ time and $ \widetilde{O}(({k^*}) n +|B|)$ space..
\end{restatable}
\begin{proof}

We analyse the expected running time and space of the algorithm.

To do so, we bound the running time of the main loop. In iteration $j$ we have $k=2^j$. Let  $j^*\in\NN$ be the iteration such that $2^{j^*-1}<12 k^*\leq 2^{j^*}$. We show that the algorithm terminates at latest in iteration $j^*$.
Let $t$ be the uncovered point chosen in any iteration of the \textsc{kApproxCover} algorithm during iteration $j^*$ of the \textsc{ApproxCover} algorithm. Since $H$ is a structured $\alpha\Delta$-covering of $S$, there has to be a element $h\in H$ such that $t\in\Psi_{\alpha\Delta}^{i,j}(S,h)$. If \textsc{WeightUpdate} is called in this iteration, then the weight of $h$ gets doubled. Let $z_{h,i}$ be the number of times the weight of $h$ has been doubled after $i$ calls of \textsc{WeightUpdate}. It is $k=2^{j^*}$. So after $i$ calls, we have \[w(H)=\sum_{h\in H} 2^{z_{h,i}}\geq 12 k^* 2^{\frac{i}{12 k^*}}\geq \frac{k}{2} \cdot 2^{\frac{i}{k}}.\] 
We always update the weight of a set $F\subset B$ with $w(F)\leq \frac{1}{r}w(B)=\frac{1}{2k}w(B)$. So after iteration $i$, we have \[w(B)\leq |B|(1+\frac{1}{2k})^i\leq |B|e^{\frac{i}{2k}}.\] Since $H$ is contained in $B$ it holds that $w(H)\leq w(B)$. Therefore it also holds 
\[\frac{k}{2}2^{\frac{i}{k}}\leq w(H)\leq w(B)\leq |B|e^{\frac{i}{2k}}\leq |B|e^{\frac{3i}{4k}}.\]
This directly implies $i< 5k \log_2(\frac{|B|}{k})$. So the algorithm has to terminate at latest in iteration $j^*$ before reaching $i=i_{max}$. The calculation of the bound above is analogous to the calculation of a similar bound in \cite{bronnimann1995almost}. It is contained here for the sake of completeness. To bound the running time in each iteration of the main loop, we claim the following.

\begin{claim}\label{clm:runtime:basic}
 A call to \textsc{kApproxCover} with inputs $S,B,r,\Delta',k'$ and $i_{max}$ has an expected running time of $O(i_{max}(k'\log(|B|)+mk'\log(k')+|B|))$ and needs $O(mk'+|B|)$ space.
\end{claim}

We postpone the proof of the claim until the end of the proof of the theorem and first finish this proof based on the statement of the claim.

In iteration $j$, we have $i_{max}=5k\log_2(\frac{|B|}{k})=5\cdot 2^j\log_2(\frac{|B|}{2^j})$. By Claim~\ref{clm:runtime:basic} the running time of \textsc{kApproxCover} in iteration $j$ can therefore be written as $2^j\cdot \text{Poly}(2^j,d,|B|,|E(S)|)$, where $\text{Poly}(2^j,d,|B|,|E(S)|)$ is a polynomial in $2^j,d,|B|$ and $|E(S)|$. Since the loop terminates at latest in iteration $j^*$, the expected running time of the loop can be bounded by
\[\sum_{j=1}^{j^*} 2^j\cdot \text{Poly}(2^j,d,|B|,|E(S)|)\leq 2^{j^*+1}\cdot \text{Poly}(2^{j^*},d,|B|,|E(S)|).\]
Because $2^{j^*}<24 k^*$, this running time is in $O( k^* \cdot \text{Poly}(k^*,d,|B|,|E(S)|))$. Inserting the polynomial, we get \[O( k^*\log(|B|/k^*)(k^*d\log(k^*d)\log(|B|)+|E(S)| k^*d\log( k^*d)^2\log(\log( k^*d))+|B|))\] 
which is $\widetilde{O}(n{k^*}^2d+|B| k^*)$ since $|E(S)| \in O(n)$. The required space for any execution \textsc{kApproxCover} is also dominated by the required space in the last execution, which is by Claim~\ref{clm:runtime:basic} in $\widetilde{O}(n {k^*}d+|B|)$.

Since the algorithm terminates in iteration $j^*$ and $2^{j^*}<24 k^*$, the size of the output $C\subset B$ of the algorithm is in $O(k^* \log(k^*))$. To be output by the algorithm, the set $C$ has to be an structured $\alpha\Delta$-covering of $S$. This is the case, because otherwise the subroutine \textsc{PointNotCovered} would have found a point that is not in the structured $\alpha\Delta$-covering of $S$ and the algorithm would not have terminated. By Observation~\ref{obs:unstructure_cover}, the output $C$ is also a $\alpha\Delta$-covering of $S$. Since $S$ is a $\Delta$-good simplification of $P$, we have $d_F(S,P)\leq 3\Delta$. So by Lemma \ref{lem:covertransfer}, $C$ is a $(\alpha+3)\Delta$-covering of $P$.

\end{proof}

\begin{proof}[Proof of Claim~\ref{clm:runtime:basic}]
By Lemma~\ref{cor:vccand}  the VC-dimension of $(B,\RSpace)$ is in $O(d)$. In Section~\ref{sec:feas}  we made this result more precise and show in Lemma~\ref{lem:vcdim:const} that the VC-dimension of $(B,\RSpace)$ is at most $\gamma=110d+412$. So, we get from Theorem~\ref{thm:epsnet} that for $k'=\lceil 8\gamma rd\log(8\gamma rd)\rceil$ and $i<i_{max}$ any draw of a set $C$ from $\Dist_i$  with $|C|=k'$ has a probability of at least $\frac{1}{2}$ to be a $\frac{1}{r}$-net of $B$. By the definition of the $\frac{1}{r}$-net and $\Dist_i$, we have for each $\frac{1}{r}$-net $N$ of $B$ and each $t\in [0,1]$ which is not covered by $N$ that $\Pr[\Dist_i]{F} \leq \frac{1}{r}$. Therefore in any iteration of the main loop the counter $i$ increases with a probability of at least $\frac{1}{2}$. So the main loop has in expectation $O(i_{max})$ iterations. 
As described in Section~\ref{sec:descr}, the cumulative probability distribution $\Dist_i$ is stored in an array that needs $O(|B|)$ space. Therefore $k'$ elements can be sampled from $\Dist$ in a total time of $O(k'\log(|B|))$ by a binary search on the array for a given random number. By Lemma~\ref{lem:findpointstructured} the subroutine \textsc{PointNotCovered} takes $O(mk'\log(k'))$ time and $O(mk')$ space. A call to the function \textsc{IsFeasible}  takes constant time and space, since $|J| \in O(1)$, and since by Lemma~\ref{lem:tdist} each check of whether $t$ is in $\Psi^{i,j}_{\Delta'}(S,e)$  can be done in $O(|j-i|)=O(1)$ time and space, if the necessary pointers are provided. Therefore, computing the set $F$ during one iteration of the main loop takes time in $O(|B|)$. 

While adding elements to $F$ we simultaneously keep track of the weight of $F$ and $B$, so that the evaluation of $\Pr[\Dist_{i}]{F} \leq \frac{1}{r}$ takes constant time. We also simultaneously construct an array of the updated probability distribution $\Dist_{i+1}$ for the case that $\Pr[\Dist_{i}]{F} \leq \frac{1}{r}$ and the subroutine \textsc{WeightUpdate} is called. This update can then be done by just switching the pointers of the two arrays in $O(1)$ time and an additional $O(|B|)$ space.

\end{proof}

\begin{proof}[Proof of Theorem~\ref{thm:runtime:main:canon}]
By Theorem~\ref{thm:runt:simpl}, there exists an algorithm that computes a $\Delta$-good simplification $S$ of $P$ in $O(n\log^2(n))$ time and $O(n)$ space.  Given this simplification, there exists by Theorem~\ref{thm:alg_candidates} an algorithm that constructs in $O(n^3)$ time and space a set of candidates $B \subset\CandidatesArg{E(S)}\subset \XX^d_2$ with $|B| \in O(n^3)$, such that $B$ contains a structured $8\Delta$-covering $C_B$ of $S$ of size at most $12k$.  If we use these two algorithms to generate $S$ and $B$, then the statement of Theorem~\ref{thm:runtime:main:canon} follows immediately from Lemma~\ref{lem:runtime:main:basic}.
\end{proof}

\subsection{Results for c-packed polygonal curves}\label{sec:analysis:cpacked}

If the input curve is a $c$-packed polygonal curve, we can obtain even better results that depend on $c$. The first two lemmas and proofs of this section are reminiscent of Lemma 4.2 and Lemma 4.3 in \cite{driemelHW12}. The main difference is the definition of the simplifications used. The third lemma uses the second lemma to show that the number of generating triples is in $O(n c^2)$ for $c$-packed curves improving upon the naive bound of $O(n^3)$ in the worst case. Lastly we talk about how to compute this set in an output-sensitive manner.

\begin{restatable}{lemma}{simpacked}
\label{lem:simpacked}

    Let $X$ be a curve in $\RR^d$. Let $S$ be a simplification of $X$, with $d_F(X,S)\leq\Delta$. Then $\|X\cap \disk[r+\Delta]{p}\|\geq\|S\cap\disk[r]{p}\|$ for any ball $\disk[r]{p}$.
\end{restatable}
\begin{proof}
    Consider any segment $u$ of the simplification $S$ that intersects $\disk[r]{p}$, defined by vertices $p_i$ and $p_j$ of $P$, with $P(t_i)=p_i$ and $P(t_j) = p_j$, and let $v=u\cap \disk[r]{p}$. Observe, that $P[t_i,t_j]$ lies inside a capsule of radius $\Delta$ around $u$. Now erect two hyperplanes passing through the endpoints of $v$. $P[t_i,t_j]$ must intersect both, hence the length of $P[t_i,t_j]$ inside a capsule of radius $\Delta$ around $v$ is at least $||v||$. As this capsule lies completely inside $\disk[r+\Delta]{p}$, the claim follows.
\end{proof}

\begin{restatable}{lemma}{driemeleins}
\label{lem:driemel1}
    Let $X$ be a given $c$-packed curve, and $S$ a $\Delta$-good simplification of $P$. Then $S$ is $54c$-packed.
\end{restatable}

\begin{proof}
    Let $\mu=d_F(X,S)$ and observe that $\mu\leq3\Delta$. Assume for the sake of contradiction, that $\|S\cap b(p,r)\|>54cr$ for some $b(p,r)$. If $r\geq\mu$ set $r'=2r$. Then by Lemma \ref{lem:simpacked} together with our assumption
    
    $$||X\cap b(p,r')||\geq ||X\cap b(p,r+\mu)||\geq ||S\cap b(p,r)||>54cr>6cr=3cr'.$$
    This contradicts the fact that $X$ is $c$-packed.
    If $r<\mu$ let $U$ denote the segments of $S$ intersecting $b(p,r)$ and let $k=|U|$. $k>54cr/2r=27c$.
    
    $$||S\cap b(p,2\mu)||\geq ||S\cap b(p,r+\mu)||\geq ||U\cap b(p,r+\mu)||\geq k\Delta/3 \geq k\mu/9=3c\mu$$
    
    since every segment of $B$ has minimal length $\Delta/3$. Hence by Lemma \ref{lem:simpacked}
    
    $$||X\cap b(p,3\mu)||\geq ||S\cap b(p,2\mu)||> 3c\mu,$$
    again contradicting the fact, that $X$ is $c$-packed.
\end{proof}

\begin{lemma}\label{lem:key-c-packed}
    Let $S$ be a $\Delta$-good simplification of some $c$-packed curve $P$. Then $$|\mathcal{T}|=O\left(nc^2\right).$$ 
\end{lemma}
\begin{proof}
    Consider the set of generating triples $\mathcal{T}$ described in Definition~\ref{def:triples}. Call the length $\|\cdot\|$ of a curve the sum of lengths of its edges. Note that for any container $a$ that is part of some triple in $\mathcal{T}$, $||a||\geq\Delta/3$ by Definition \ref{def:goodsimp}.
    
    Now we associate every triple $(e,c_1,c_2)\in\mathcal{T}$ to the shortest of $e,c_1$ or $c_2$. As $e$ can also be interpreted as a $1$-container, we will call every element of a generating triple a container in this proof.
    
    Consider some container $a$ and all its associated triples. We want to bound the number $k$ of containers that can be part of associated triples by $O(c)$.
    
    Find some enclosing ball $\disk[||a||]{p}$ of $a$. Since for every container $b$ that can participate in an associated triple of $a$, only if there are two points $p_1\in a$ and $p_2\in c$ with $||p_1-p_2||\leq16\Delta$, $\disk[||a||+16\Delta]{p}$ contains at least one point of all $k$ associated other containers. Since $a$ is shorter than any other associated containers, $\disk[2||a||+16\Delta]{p}$ contains a section of all other containers, of length at least $||a||$. Furthermore, note that any edge of $S$ is in at most $6$ containers. Thus $k$ is bounded by 
    
    \begin{align*}
        k&\leq \frac{6||S\cap\disk[2||a||+16\Delta]{p}||}{||a||}\leq\frac{324c(2||a||+16\Delta)}{||a||}\\
        &= \frac{648c||a||}{||a||} + \frac{5184c\Delta}{||a||}\leq648c + \frac{15552c\Delta}{\Delta}\leq 16200c = O(c),
    \end{align*}
    where the second inequality follows, as $S$ is $54c$-packed by  Lemma~\ref{lem:driemel1}.
    Hence there are $O(c^2)$ triples associated to $a$, as any pair of participating containers could form a triple with $a$.
    Summing over all $O(n)$ containers implies the claim.
\end{proof}

Note that the above proof does not imply that a single container takes part in only $O(c^2)$ triples. Indeed, there may be a very long container that takes part in many triples, but all triples are associated to the shorter ones.

\begin{lemma}\label{lem:cpackedcandidates}
Let $S$ be a $\Delta$-good simplification of some $c$-packed curve $P$ in $\RR^2$. Then the set of generating triples $\mathcal{T}$ can be computed in $O(nc^2 \log(n))$ time and space.
\end{lemma}

\begin{proof}
By Lemma~\ref{lem:key-c-packed} the number of generating triples is bounded by $O(nc^2)$. We compute the set of triples as follows. We first find all close edge-edge pairs. Since $S$ is $O(c)$-packed these are only $O(nc)$ many. This can be done with an output-sensitive intersection algorithm in $O(nc\log(nc))$ time, computing the arrangement of edges, and boundaries of $\Delta$ neighbourhoods of edges. Now similarly to the proof of Lemma \ref{lem:key-c-packed}, associate every pair to the shorter of the two edges, resulting in $O(n)$ sets of pairs with $O(c)$ members at most. From this set, we form $O(nc^2)$ triples $T'$ of edges. For every triple $t$ in $T'$, we add every permutation of $t$ to $T'$. Note that for every triple $(e,c_1,c_2)$ in $T_S$ there are edges $e_1\in c_1$ and $e_2\in c_2$, such that $(e,e_1,e_2)$ is in $T'$. This follows from the fact that there is a shortest edge among $e,e_1$ and $e_2$, and hence the triple would have been added to $T'$. As all containers in any triple of $T_S$ consists of at most $3$ edges, for any $(e,e_1,e_2)\in T'$ form all possible triples $(e,c_1,c_2)$ such that $e_1\in c_1$ and $e_2\in c_2$, with $c_1$ and $c_2$ containing at most three edges. This results in a superset of $T_S$ of size $O(nc^2)$ computed in $O(nc\log (n) + nc^2)$ time. Note that duplicates can be removed in $O(nc^2\log(n))$ time. Similarly, the condition for a generating triple can be checked in $O(1)$ time, thus removing all triples, that are not generating triples, takes $O(nc^2)$ time.
\end{proof}

\begin{observation}
    In $\RR^d$ we can compute the set of generating triples $\mathcal{T}$ in $O(n^2 + nc^2)$ time, by first na\"ively finding all close edge pairs ( $O(n^2) time$ and $O(nc)$ space) and for each edge marking all close subcurves via this approach. Afterwards we form all triples, which takes $O(|\mathcal{T}|)=O(nc^2)$ time and space.
\end{observation}

From Lemma~\ref{lem:cpackedcandidates}, Lemma \ref{lem:runtime:main:basic}, and Theorem~\ref{thm:alg_candidates} the following result is immediate, since the subset of candidates that results from a fixed triple can be computed in $O(1)$.

\begin{restatable}{theorem}{cpackedresult}
\label{thm:cpackedresult}
    Let $P\in\XX^d_n$ be $c$-packed and $\Delta\in\RR_+$. Let $k^*$ be the minimum size of a solution to the $\Delta$-coverage problem on $P$. There exists an algorithm that outputs an $(11,O(\log(k^*)))$-approximate solution. The algorithm needs
    \begin{enumerate}
        \item $\Tilde{O}((k^*)^2n + n c^2 k^*)$ expected time and $\Tilde{O}(k^*n + nc^2)$ space in $\RR^2$,
        \item $\Tilde{O}((k^*)^2n + n c^2 k^* + n^2)$ expected time and $\Tilde{O}(k^*n + nc^2)$ space in $\RR^d$.
    \end{enumerate}
\end{restatable}

We expect that the running time of the algorithm in $\RR^d$ can be improved via an explicit construction of the arrangement of suitable linear approximations of $\Delta$-neighbourhoods of edges of the simplification $S$ of $P$. As $P$ is $c$-packed, it is conceivable that this can be done such that the arrangement complexity can be bounded by $O(\poly(c)\cdot n)$. However, constructing the arrangement in an output-sensitive manner requires more care and we leave this for future work.

\section{Algorithm variant with implicit weight update}\label{sec:implicit:weights}

In this section we describe a variant of Algorithm~\ref{alg:main} that will lead to better running times and space requirements in terms of $n$, the number of vertices of the input curve, at the expense of introducing a (polylogarithmic) dependency on the arclength of the input curve.

Instead of defining a function  \textsc{GenerateCandidates} to compute the set $B$ explicitly, we will maintain the discrete probability distribution $\Dist_i$ of Algorithm~\ref{alg:main} on an implicitly defined approximate candidate set. To this end, we introduce the following definition of a candidate set which we implicitly generate from the edges of the simplification $E(S)$.

\begin{definition}[$\eps$-approximate candidate set] \label{def:eps:candidates}
Let $E$ be a set of edges in $\RR^d$ and let $\eps > 0$ be a parameter. We can approximate the candidate space induced by $Q$ by subsampling the edges as follows. Define $G_{\delta}:=\{i\cdot \delta \mid i\in \ZZ\}$. For each edge $e_i \in E$ consider the set 
\[ X_i = \left( [0,1] \cap G_{\eps/\lambda_i} \right) \cup \{ 1 \} \]
where $\lambda_i$ is the length of the edge $e_i$. The $\eps$-approximate candidate set induced by $E$ is the set 
\[ \CandidateSet{\eps,E} = \bigcup_{i=1}^{m} \{ (x,y,i) \mid x \in X_i, y \in X_i \}\]
Assuming $E$ fixed, observe that for any edge $p \in \CandidatesArg{E}$, there exists an edge $p' \in Z_{\eps,E}$, such that $d_F(e,e')\leq \eps$. 
\end{definition}

In the following, we use the fact that the feasible set restricted to this candidate set has a nice structure that can be stored implicitly and can be computed fast. In particular, Lemma~\ref{lem:rect} states that the feasible set, when restricted to the subedges of an edge $e \in E$, can be written as the union of constantly many rectangles.
This structure is illustrated in Figure~\ref{fig:feasset}. The set system of feasible sets $ \{F_{\Delta} (t) \mid t \in \Param{n} \} $ was defined in Section~\ref{sec:overview:ana:main}.

\begin{figure}[t]
    \centering
    \includegraphics[width=0.9\textwidth]{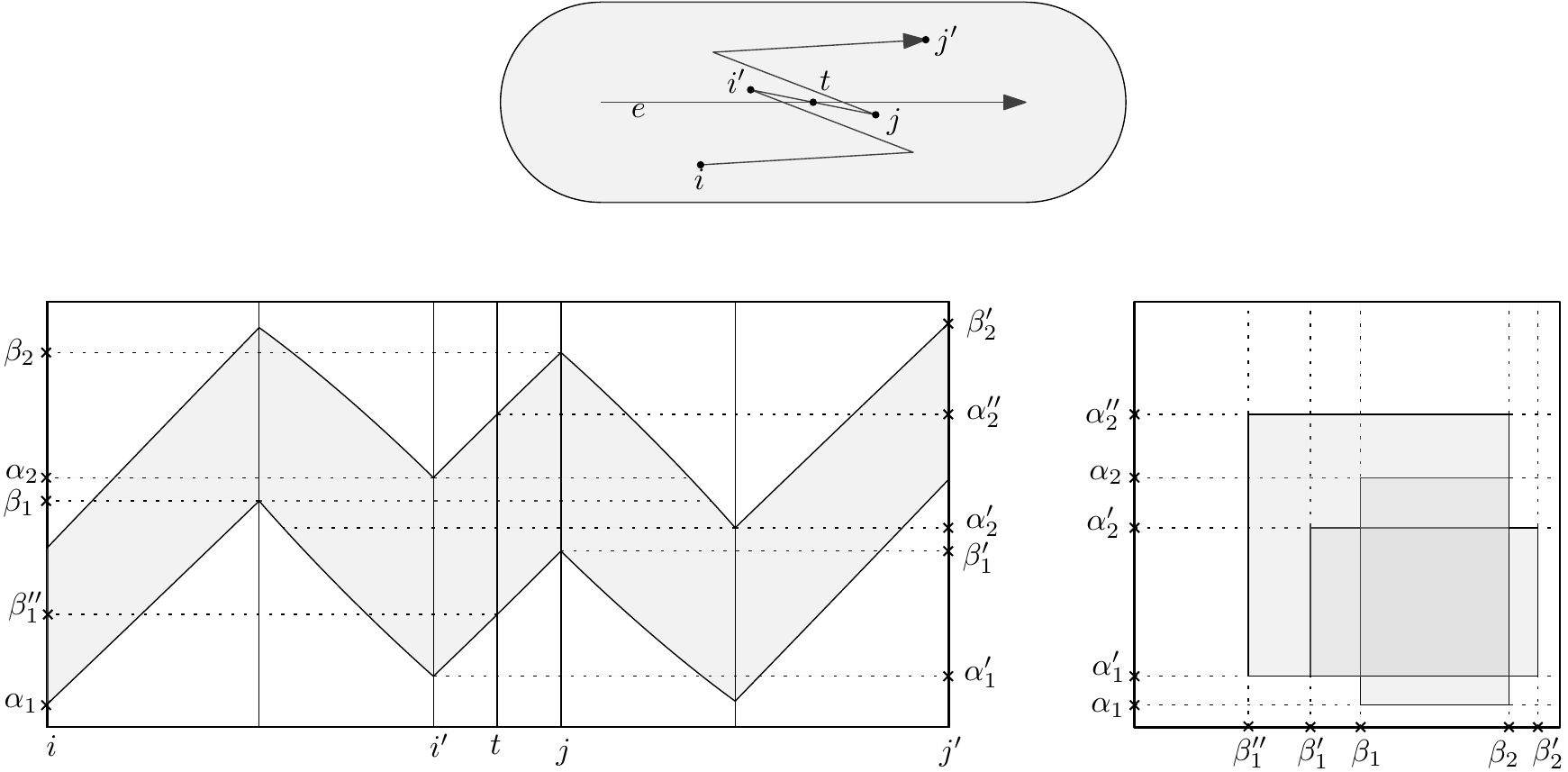}
    \caption{Structure of a feasible set. The image shows an example of the structure of the feasible set $F_\Delta(t)$ restricted to the subedges of some edge $e\in E$ for some curve $S$. The top image displays the edge $e$ and the subcurve $S[t_i,t_{j'}]$ that includes all subcurves of $S$ with complexity $4$ that contain $t$. The left image shows the $\Delta$-freespace diagram of $e$ and $S[t_i,t_{j'}]$ with notations of the relevant parameters for computing the corresponding rectangles in the parameter space of $e$ based on Lemma~\ref{lem:rect}. The right image shows the feasible set as a union of these rectangles. The rectangle $[\alpha_1,\alpha_2]\times[\beta_1,\beta_2]$ corresponds to the tuple $(i,j)$, the rectangle $[\alpha'_1,\alpha'_2]\times[\beta'_1,\beta'_2]$ to $(i',j')$ and $[\alpha'_1,\alpha''_2]\times[\beta''_1,\beta_2]$ to $(i',j)$. All other tuples correspond to rectangles that are contained in these three rectangles.}
    \label{fig:feasset}
\end{figure}

\subsection{Data structure for sampling}\label{sec:ds}

We describe a simple static data structure to store the discrete probability distribution $\Dist_i$ over the finite set $Z_{\Delta,E(S)}$, which we use in our adaptation of Algorithm~\ref{alg:main}. The data structure takes as input a polygonal curve $S \in \XX^{d}_n$, and a sequence of values $t_1,\dots,t_i \in [0,1]$, which are used for the weight update.  The data structure should support the following operations.
\begin{enumerate}[(1)]
    \item Sample and return an (explicit) element of $Z_{\Delta,E(S)}$ according to $\Dist_i$.
    \item Given a query point $t$, determine if $\Pr[\Dist_{i}]{F_{\Delta}(t)}$ is at most $\frac{1}{r}$. 
\end{enumerate}

\subparagraph{The data structure}
For each edge $e$ of $E(S)$, we will store an arrangement $\mathcal{A}_e$ of the horizontal and vertical lines that delineate the boundaries of the rectangles that form the sets $F_{\Delta}(t_1),\dots F_{\Delta}(t_i)$ when restricted to the parameter space of subedges of $e$. For each cell $C$ of this arrangement, we store the following information (refer to Figure~\ref{fig:arrangement:feas} for an illustration): 
\begin{enumerate}[(i)]
\item $g_C$, the number of grid points in the cell $|C \cap Z_{\Delta,e} |$
\item $s_C$, the number of feasible sets $F_{\Delta}(t_1),\dots F_{\Delta}(t_i)$ that contain $C$
\end{enumerate}

The weight function $w:Z_{\Delta,E(S)} \rightarrow \RR_{+}$ that defines the distribution $\Dist_i$ can then be evaluated on the cell $C$ as 
\[ w(C) = 2^{s_C} g_C\]

In particular, this weight function defines the probability distribution $\Dist_i$ in the following sense. The probability of a grid point $Q \in Z_{\Delta,E(S)} \cap C$ is given by \[\Pr{Q} = 2^{s_C}/w(Z_{\Delta,E(S)}) \]

For computing the set of cells, for each edge $e \in E(S)$, we build the arrangement $\mathcal{A}_e$ by first collecting the coordinates of the horizontal and vertical lines that delineate the feasible sets and then sorting them by $x$ (resp. $y$)-coordinate. Since the arrangement is a (non-uniform) grid, the information for all cells can be stored in an array with appropriate indexing.
Initially, for $i=0$, we only have one cell $C$ for each edge $e_j \in E(S)$, which is the unit square, and we set $s_C=0$, and $g_C=|X_j|$ (see Definition \ref{def:eps:candidates}).
For $i>0$, we compute the number of feasible sets  $s_C$ using dynamic programming, by scanning over the arrangement of cells in a column by column fashion. 
Computing the number of grid points $g_C$ can be done by scanning over the arrangement in a similar way. Here, we compute the number of gridpoints in the interval between two horizontal or vertical lines by using a binary search on the set of gridpoints $X_j$. Clearly, computing the arrangement $\mathcal{A}_e$ and the values $s_C$ and $g_C$ for each cell can be done in $O(i^2\log(|X_j|))$ time and space per edge $e_j \in E(S)$.

Let  $\mathcal{M}=\{C_1,\dots,C_{m}\}$  denote the union of the set of cells over all arrangements of the edges of $E(S)$, using an appropriate indexing (i.e., lexicographical ordering in the horizontal, and vertical direction, and in the index of the edge $e$ of $E(S)$). 
In addition to the values $s_C$ and $g_C$, we store the cumulative function $f: \mathcal{M} \rightarrow \RR_{+}$, which is simply defined as $f(C_j) = \sum_{i=1}^{j} w(C_i)$ and can be computed by scanning over all cells in the order of their index. For consistency, we define $f(C_0)=0$.
Note that the total weight $w( Z_{\Delta,E(S)})$ is now stored in $f(C_m)$. The function $f$ can be computed in $O(m)$ time and space and this also bounds the total space used by the data structure.

\begin{figure}[t]
    \centering
    \includegraphics[width=1\textwidth]{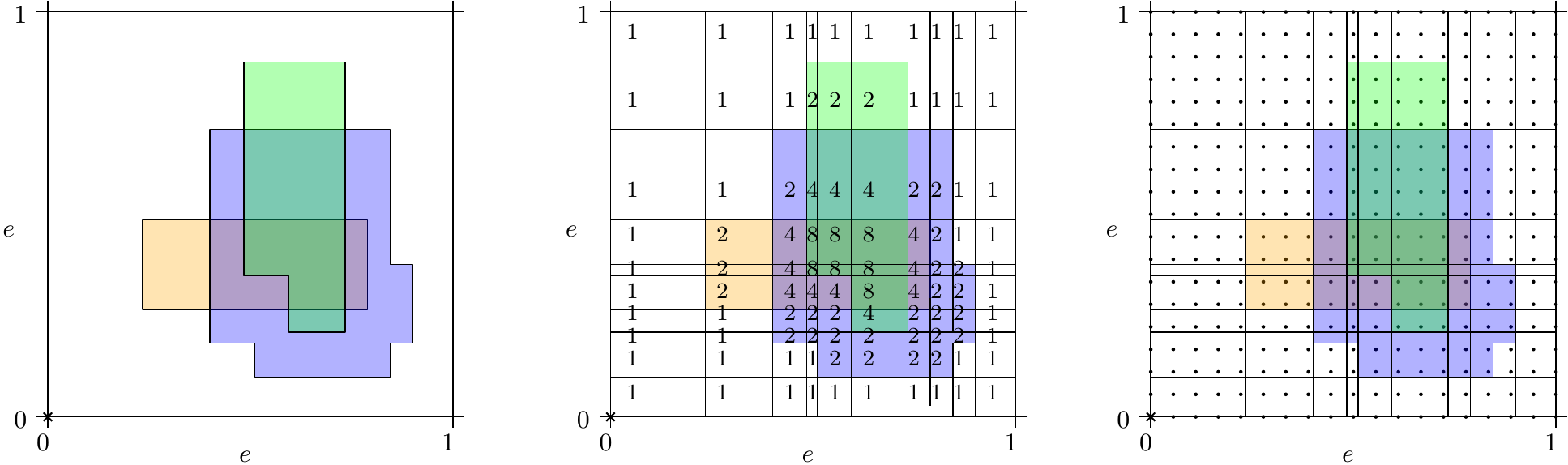}
    \caption{Left: Example of an arrangement of lines stored in the data structure for one edge $e\in S$. The left image shows the arrangement of  the feasible sets that were used for the updates. Center: The actual arrangement of lines that is stored in the data structure. The number shown in each cell is the corresponding weight multiplicator resulting from doubling the weights in each update where the cell is included in the feasible set. Right: The grid of candidates in $Z_{\Delta, E(S)}$ that are subcurves of $e$. The probability to draw a candidate is proportional to the weight multiplicator of the cell that contains it.}
    \label{fig:arrangement:feas}
\end{figure}

\subparagraph{Sampling from the distribution}
Using the cumulative function $f$ defined on the cells of the arrangements and the additional information for each cell, we can sample from $\Dist_i$ as follows. Draw a sample $x$ uniformly at random from the interval $[0,w(Z_{\Delta,E(S)})]$. Perform a binary search on the function values of the cumulative function for $x$, let cell $C_j$ be the result of the binary search. Let $j'=\lceil (x-f(C_{j-1}))/2^{s_{C_j}} \rceil$ and return the $j'$-th grid point (according to a lexicographical ordering) that lies in the cell $C_j$. Clearly, this can be done in time in $O(\log m + \log (\lambda/\Delta))$, where $\lambda$ denotes the length of the longest edge of $S$.

\subparagraph{Evaluating the weight of a feasible set}
With the data structure as described above, we can evaluate  $\Pr[\Dist_{i}]{F_{\Delta}(t)}$ as follows.
For each edge $e\in E(S)$, we find the set of cells intersected fully or partially by the feasible set $F_{\Delta}(t)$ by scanning over all cells associated with the edge $e$. If a cell is intersected only partially, we can determine the number of grid points that lie in the intersection by using a constant number of binary searches, since  $F_{\Delta}(t)$ is the constant union of a set of rectangles when restricted to the parameter space of the edge $e$. From this, we can compute the weight of the feasible set by summing over all intersected cells. Dividing this weight by the total weight $w( Z_{\Delta,E(S)})$ yields the probability $\Pr[\Dist_{i}]{F_{\Delta}(t)}$. Clearly, the total time for evaluating the weight of one feasible set is in $O(m \log (\lambda/\Delta))$. (Better running times are  possible by storing the cumulative function in a more structured way, but this does not affect our total running time.)

\vspace{\baselineskip}
We conclude the section with a theorem summarizing what we have derived.

\begin{theorem}\label{thm:datastr:runtime}
    Given a polygonal curve $S \in \XX^{d}_n$, and a sequence of values $t_1,\dots,t_i \in [0,1]$, we can build a data structure that supports the following operations:
\begin{compactenum}[(1)]
    \item Sample and return an explicit element of $Z_{\Delta,E(S)}$ according to $\Dist_i$.
    \item Given a query point $t$, determine if $\Pr[\Dist_{i}]{F_{\Delta}(t)}$ is at most $\frac{1}{r}$.
\end{compactenum}
    Let $m=n \cdot i^2$, and let $\lambda$ denote the length of the longest edge of $S$.
    The query time for (1) is in $O(\log (m) + \log (\lambda/\Delta))$. The query time for  (2) is  in $O(m \log (\lambda/\Delta))$.
    The data structure can be built in $O(m \log (\lambda/\Delta))$ time and uses space in $O(m)$.
\end{theorem}

\subsection{Result for implicit weight update}

By using the data structure of Section~\ref{sec:ds} for maintaining the discrete probability distribution on the implicit candidate set $Z_{\Delta,E(S)}$, we obtain the following result.

\begin{restatable}{theorem}{implicitweightresult}
\label{lem:runtime:main:improved}
Let $P\in \XX^d_n$ and $\Delta\in\RR_+$.
Let $k^{*}$ be the minimum size of a solution to the $\Delta$-coverage problem on $P$. Let further $\length(P)$ be the arc length of the curve $P$.
There exists an algorithm that outputs a 
$(12, O( \log  (k^{*}))$-approximate solution. The algorithm needs in expectation  $O(n{k^*}^3(\log^4(\frac{\length(P)}{\Delta })+\log^3(\frac{n}{ k^*}))+n\log^2(n))$ time and $O(n{k^*}\log(\frac{n \length(P)}{\Delta k^*}))$ space. 
\end{restatable}

\begin{proof}
     By Theorem~\ref{thm:candidate_space} and the use of triangle inequality, there exists a $9\Delta$-covering $C\subset Z_{\Delta,E(S)}$ of $S$ with $|C|\leq 3k^*$. Therefore, we get by Lemma~\ref{lem:runtime:main:basic} that the basic algorithm \textsc{ApproxCover} outputs an 
    $(12, O( \log  (k^{*}))$-approximate solution. Since we just modify the algorithm with a data structure for faster maintenance of the weighted probability distribution, the quality of the solution stays the same.
    
    It remains to bound the expected running time and space. The subroutine \textsc{SimplifyCurve}  stays unchanged and needs $O(n\log^2(n))$ time and $O(n)$ space, as shown in Theorem~\ref{thm:runt:simpl}. Let  $j^*\in\NN$ be the iteration such that $2^{j^*-1}<3k^*\leq 2^{j^*}$. Analogous to the proof of Lemma~\ref{lem:runtime:main:basic}, we see that the algorithm still terminates in iteration $j^*$ and the running time for each call of the improved version of \textsc{kApproxCover} is dominated by its last call. The analysis for the running time of the improved version of \textsc{kApproxCover} has to be changed only for the operations that use the data structure now. It still holds that the expected number of iterations of the main loop is in $O(i_{max})$. It also still holds that by Lemma~\ref{lem:findpointstructured} the subroutine \textsc{PointNotCovered} takes $O(nk'\log(k'))$ time and $O(nk')$ space. Since the simplification $S$ computed by \textsc{SimplifyCurve} only uses vertices of $P$ as its vertices, we have that each edge of $S$ has at most length $\length(P)$. By this argument and Theorem~\ref{thm:datastr:runtime}, the time and space needed by the data structure in any iteration of the main loop is bounded by $O(n i_{max}^2\log(\length(P)/\Delta)+k'(\log(\length(P)/\Delta)+\log(n i_{max})))$ and $O(ni_{max}^2)$. In iteration $j^*$, the value $i_{max}$ is in $O(k^*\log_2(\frac{|Z_{\Delta,E(S)}|}{k^*}))$. For each edge, there are at most $\length(P)/\Delta+2$ start points and $\length(P)/\Delta+2$ end points that together define a candidate in $Z_{\Delta,E(S)}$. Therefore, the value $|Z_{\Delta,E(S)}|$ is in  $O(n\frac{\length(P)^2}{\Delta^2})$. So $i_{max}$ is in $O(k^*\log_2(\frac{n\length(P)}{\Delta k^*}))$. Since $k'$ is in $O( k^*\log(k^*))$ in iteration $j^*$, we get in total an expected running time of $O(n{k^*}^3(\log^4(\frac{\length(P)}{\Delta })+\log^3(\frac{n}{ k^*}))+n\log^2(n))$ and a space requirement of $O(n{k^*}\log(\frac{n\length(P)}{\Delta k^*}))$. 
\end{proof}

\section{Conclusions}\label{sec:conclusions}
With the algorithm variants presented in this paper, we can find bicriteria-approximate solutions to the $\Delta$-coverage problem on a polygonal curve $P$. The new algorithms  improve upon previously known algorithms for the $\Delta$-coverage problem both in terms of known running time and space requirement bounds \cite{akitaya2021subtrajectory}, as well as approximation factors.  To the best of our knowledge, our candidate generation leads to the first  the first strongly polynomial algorithm for subtrajectory clustering under the continuous Fr\'echet distance that does not depend on the relative arclength $\lambda/\Delta$ of the input curve or the spread of the coordinates. The running time is at most cubic in $n$, the number of vertices of the input curve (Theorem~\ref{thm:runtime:main:canon}). In practice, we expect this to be lower as testified by our analysis for $c$-packed curves (Theorem \ref{thm:cpackedresult}). The work of Gudmundsson et. al. \cite{gudmundsson2020} suggest that in practice most curves are $c$-packed for a $c$ that is considerably smaller than the complexity of the curve. However, it remains to be seen if this also holds for the typically long curves which appear as input in the subtrajectory clustering setting.
We also present a variant of the algorithm with implicit weight updates which achieves a linear dependency on $n$ (Theorem~\ref{lem:runtime:main:improved}) and this holds in general, without any $c$-packedness assumption on the input.

There are several avenues for future research. We mention some of them here.
An interesting question that remains open for now is whether the implicit weight update can be performed directly on the  candidate set (Definition~\ref{def:candidates}). For this, we need to develop a dynamic  data structure that can maintain the  distribution on this candidate set to perform updates with rectangles and to sample from it.
Another future research direction is to improve the dependency of the approximation factor on the parameter that controls the complexity of the input curves. Currently, the dependency is linear, and we did not try to improve it, since our focus was on clustering with line segments.
Another interesting question is, how the low complexity center curves obtained by our algorithm can be best connected to center curves of higher complexity or even a geometric graph while retaining the $\Delta$-covering.

The $\Delta$-coverage problem is an approach to subtrajectory clustering based on minimizing the number of centers needed to fully cover the given trajectory. A closely related problem that directly comes to mind in this context is the maximum coverage problem that consists of maximizing the fraction of the trajectory that can be covered by a fixed number of center curves. We can directly apply our candidate generation (Definition $19$) to this problem and obtain a tricriteria approximation as follows.  We use the generic greedy algorithm for maximizing submodular functions with a cardinality constraint \cite{Nemhauser78} on the candidate set (Definition $19$) to approximate the coverage of a simplification. 
The generic greedy algorithm  selects in each step the candidate curve that covers the largest amount of uncovered fraction of the input trajectory until the fixed number of centers is reached. 
It would be interesting to find faster approximation algorithms for the problem with similar or better approximation factors.

\bibliographystyle{plainurl}
\bibliography{bibliography}
\appendix

\newpage

\section{On approximating free space intervals and extremal points}\label{sec:squareroots}

In this section we want to show, that we can execute the above described algorithm in a computational model that does not permit square roots. These are used whenever we compute points in the $\Delta$-free space, as in the $2$-norm setting we implicitly solve quadratic equations. That is we compute the intersection of either disks or cylinders in $\RR^d$ with line segments, which are the underlying events for extremal points or free-space intervals. To this end we show, that we can approximate these intersections sufficiently well in $O(\mathrm{poly}(d))$ time.

\begin{observation}\label{obs:halfspaceintersection}
    Given a halfspace $H$ and an edge defined by points $p$ and $q$ in $\RR^d$, we can compute the intersection of these two sets in $O(d)$ simple operations.
\end{observation}

\begin{figure}[b]
    \centering
    \includegraphics[width=0.6\textwidth]{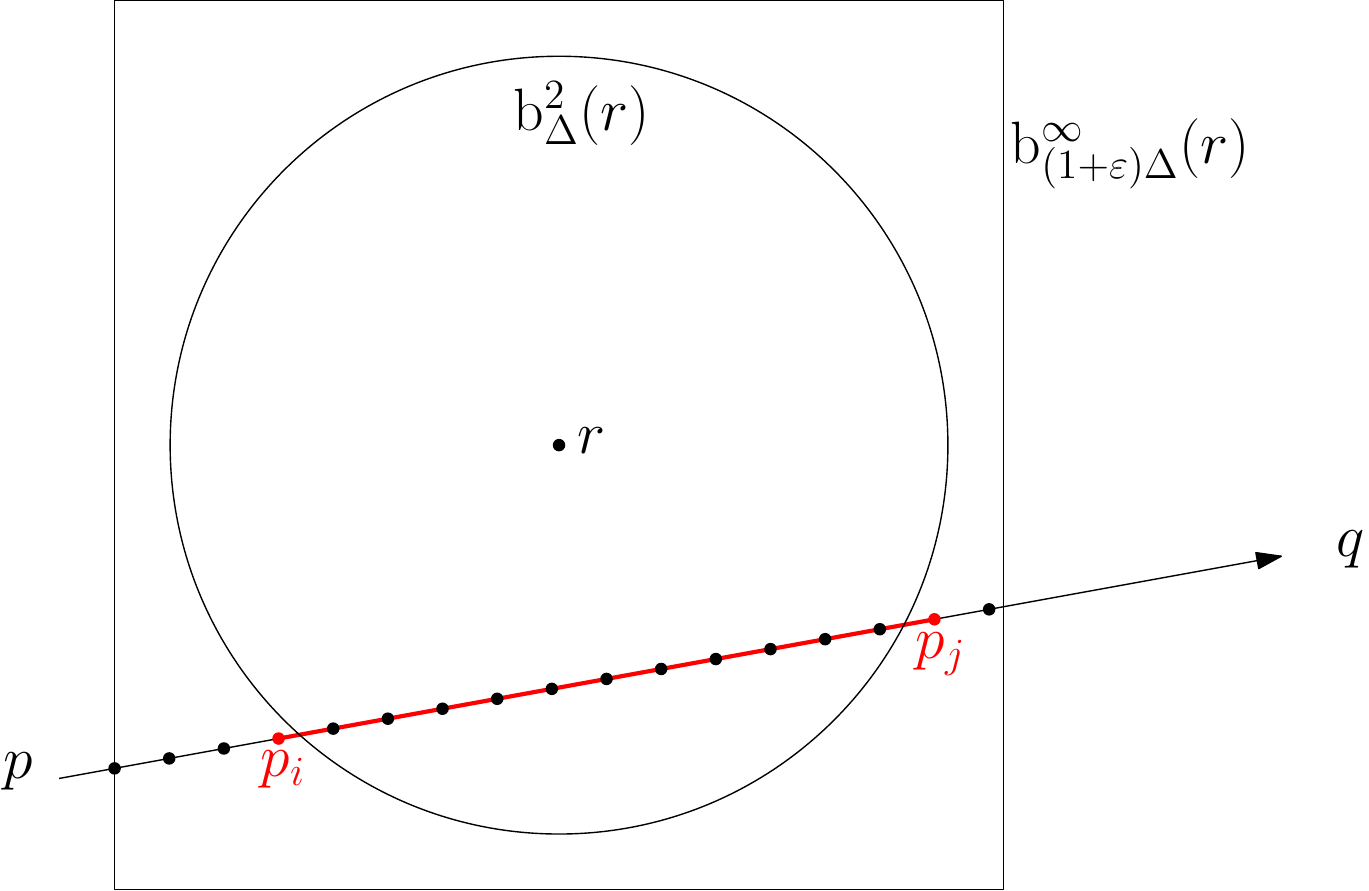}
    \caption{Illustration to Lemma \ref{lem:approxinterval}. $\overline{p_i\,p_j}$ is the computed interval approximating the intersection of the $2$-disk with $\overline{p\,q}$.}
    \label{fig:approxinterval}
\end{figure}

\begin{lemma}[approximating interval]\label{lem:approxinterval}
    Given an edge defined by points $p$ and $q$, and a point $r$ in $\RR^d$, together with values $\Delta,\eps>0$. We can compute points $a,b$ along $\overline{p\,q}$, such that
    \[\{z\in\overline{p\,q}\mid\|z-r\|_m\leq\Delta\}\subset\overline{a\,b}\subset\{z\in\overline{p\,q}\mid\|z-r\|_m\leq(1+\eps)\Delta\}\]
    with $O(d^2+d\log(1/\eps))$ simple operations.
\end{lemma}
\begin{proof}
    Refer to Figure \ref{fig:approxinterval}.
    For this proof, let $\|\cdot\| = \|\cdot\|_m$ unless otherwise stated.
    We begin by finding the set of points $\{z\in\overline{p\,q}\mid\|z-r\|_\infty\leq\Delta\}=\overline{\hat{a}\,\hat{b}}$ defined by the intersection of $\mdisk[(1+\eps)\Delta]{\infty}{r}$, with the edge $\overline{p\,q}$. In the case, that the intersection is empty, we are done, as $\|\cdot\|\leq\|\cdot\|_\infty$. Finding these intersections points is the same as finding the intersection of $O(d)$ halfspaces with the edge. Hence it takes $O(d^2)$ simple operations by Observation \ref{obs:halfspaceintersection}. We now need to refine the subedge $\overline{\hat{a}\,\hat{b}}$. For this we (implicitly) place points along $\overline{\hat{a}\,\hat{b}}$ with distance $\eps\Delta$. These are $O(1/\eps)$ many points, as $\|\hat{a} - \hat{b}\|\leq2\sqrt{2}(1+\eps)\Delta$. Call these $p_i$. We then find $i$ and $j$, such that $\|z-p_i\|>\Delta>\|z-p_{i+1}\|$ and $\|z-p_j\|>\Delta>\|z-p_{j-1}\|$. Note that these points exist, as the first and last point lie outside $m$-disk. This takes $O(d\log(1/\eps))$ simple operations, via binary search. Note that $\{z\in\overline{p\,q}\mid\|z-r\|_m\leq\delta\}\subset\overline{p_i\,p_j}$. Also note, that $\|z-p_i\|\leq\|z-p_{i+1}\|+\|p_i-p_{i+1}\|\leq(1+\eps)\Delta$, similarly for $j$. Hence $p_i$ and $p_j$ are the sought after points, computed in $O(d^2+d\log(1/\eps))$ simple operations, as $\|\cdot\|$ is a convex function.
\end{proof}

\begin{lemma}\label{lem:approx-lineinterval}
Given two edges $\overline{s\,t}$ and $\overline{p\,q}$, and let $\mathrm{dist}^m_{\overline{s\,t}}(x)$ be the minimal $m$-distance between any point $x$ and the edge $\overline{s\,t}$ in $\RR^d$. Then for given $\eps,\Delta>0$ we can compute points $a,b$ along $\overline{p\,q}$, such that 
\[\{z\in\overline{p\,q}\mid\mathrm{dist}_{\overline{s\,t}}(z)\leq\Delta\}\subset\overline{a\,b}\subset\{z\in\overline{p\,q}\mid\mathrm{dist}_{\overline{s\,t}}(z)\leq(1+\eps)\Delta\}\]
with $O(d^3 + d\log(1/\eps))$ simple operations.
\end{lemma}
\begin{proof}
    Without loss of generality assume, $s=0$.
    As always let $\|\cdot\|=\|\cdot\|_m$.
    We begin by computing a orthogonal basis $V=\{v_1,\ldots,v_d\}$ of $\RR^d$, where $v_1=t$. This can be done with the Gram-Schmidt process in $O(d^3)$ operations. In $O(d^2)$ simple operations we can then transform every point into the new base $V$. The rest of this proof is done with respect to the base $V$. For any point $x$, denote by $x_1$ the first coordinate, and $x'$ the projection of $x$ ignoring the first coordinate.
     
    Then for any point $x\in\RR^d$, we know that
    \[
    \mathrm{dist}_{\overline{s\,t}}(x) =  \begin{cases}
                                                \|x-s\| & \text{if $x_1 \leq 0$}\\
                                                \|x'\| & \text{if $0\leq x_1 \leq t_1$}\\
                                                \|x-t\| & \text{if $x_1 \geq t_1$.}
                                            \end{cases}
    \]
    We then split $\overline{p\,q}$ into (up to) three segments $s_1,s_2$ and $s_3$, depending on which of the above cases is fulfilled. Without loss of generality we have exactly three segments. For the first and third segment, Lemma \ref{lem:approxinterval} computes the sought after set. For the second segment, project the start and end point by ignoring the first coordinate onto $p'$ and $q'$ in $\RR^{d-1}$. Without loss of generality, $p'\neq q'$. Then use Lemma \ref{lem:approxinterval} to compute $p_1'$ and $p_2'$ for disks centered at $0$, refer to Figure~\ref{fig:cylinderporjection}. Since $p'\neq q'$, we have a trivial bijection between $s_2$ and $\overline{p'\,q'}$, such that $\mathrm{dist}_{\overline{s\,t}}(x)=\|x'\|$, with $x\in s_2$ and $x'\in\overline{p'\,q'}$. Hence mapping $p_1'$ and $p_2'$ back onto $s_2$ also computes the sought after set on $s_2$. Finally, form the union of the three computed sets (if they are disconnected, the convex hull, as the distance function is convex), resulting in the sought after $\overline{a\,b}$. The amount of simple operations is bound by $O(d^3)$ for the Gram-Schmidt process, and $O(d\log(1/\eps))$ for each application of Lemma \ref{lem:approxinterval}.
\end{proof}

\begin{figure}
    \centering
    \includegraphics[width=0.9\textwidth]{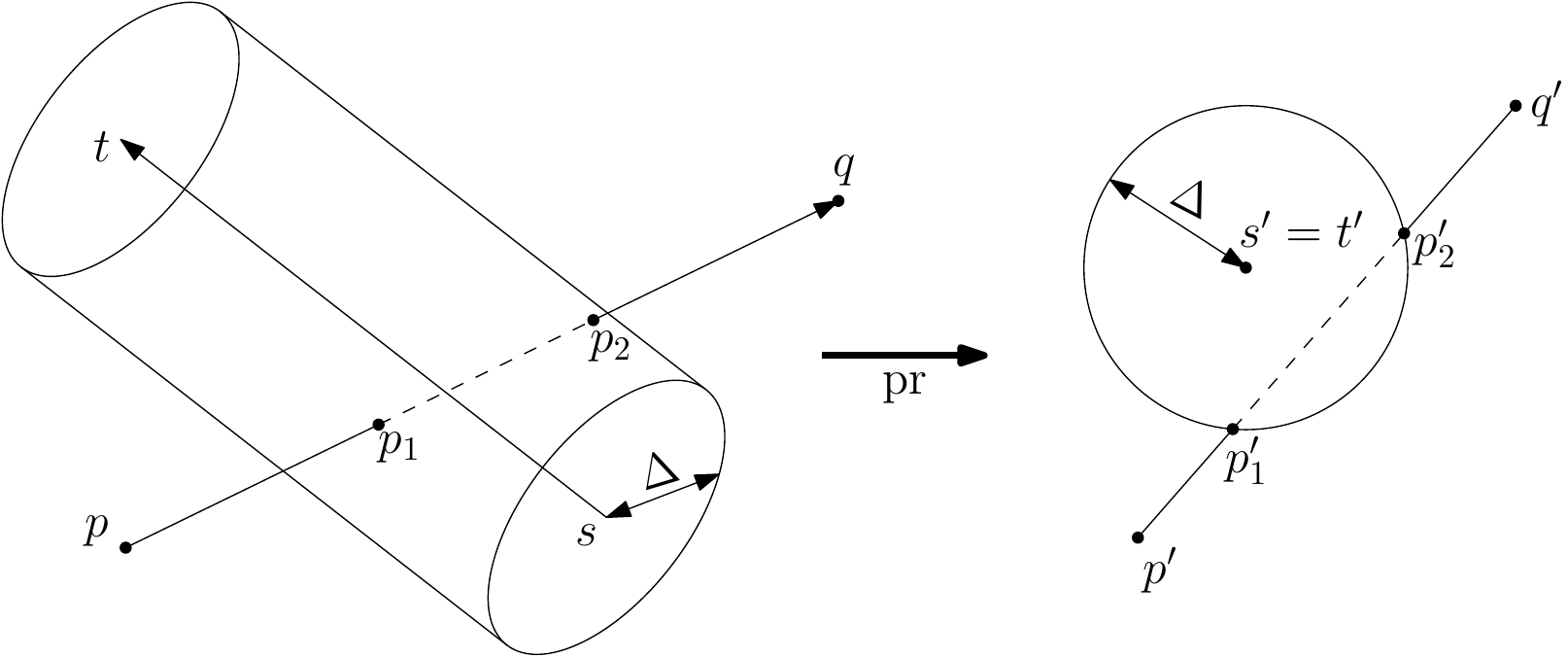}
    \caption{Projection of the cylinder in $\RR^d$ to a circle in $\RR^{d-1}$ preserving the distance function.}
    \label{fig:cylinderporjection}
\end{figure}

\newpage

\end{document}

%% file: relatedwork.tex
Buchin, Buchin, Gudmundsson, L{\"{o}}ffler and Luo were the first to consider the problem of clustering subtrajectories under the Fr\'echet distance~\cite{BuchinBGLL11}. They consider the problem of finding a single cluster of subtrajectories with certain qualities, like the number of distinct subtrajectories, or the length of the longest subtrajectory assigned to it. In their paper, they suggested a sweepline approach in the parameter space of the curves and obtain constant-factor approximation algorithms for finding the largest cluster. They also show NP-completeness of the corresponding decision problems. This hardness result extends to $(2-\eps)$-approximate algorithms. For their $2$-approximation algorithm, Buchin et al.~\cite{BuchinBGLL11} develop an algorithm that finds a legible cluster center among the subcurves of the input curve. Gudmundsson and Wong~\cite{abs-2110-15554} present a cubic conditional lower bound for this problem and show that it is tight up to a factor of $O(n^{o(1)})$, where $n$ is the number of vertices.

The algorithmic ideas presented in~\cite{BuchinBGLL11} were implemented and extended by Gudmundsson and Valladares~\cite{GudmundssonV15} who obtained  practical speed ups using GPUs.
In a series of papers, these ideas were also applied to the problem of reconstructing road maps from GPS data~\cite{buchinGroup2017,BuchinBGHSSSSSW20}.  
In a similar vain, Buchin, Kilgus and  Kölzsch~\cite{buchinGroup20} studied the trajectories of migrating animals and defined so-called group diagrams which are meant to represent the underlying migration patterns in the form of a graph. In their algorithm, to build the group diagram, they repeatedly find the largest cluster and remove it from the data, inspired by the classical greedy set cover algorithm.

The above cited works however do not offer theoretical guarantees when used for computing a clustering of subtrajectories, nor do they explicitly formulate a clustering objective. 
Agarwal, Fox, Munagala, Nath, Pan, and Taylor~\cite{agarwal2018} define an objective function for clustering subtrajectories based on the metric facility location problem, which consists of a weighted sum over different quality measures such as the number of centers and the distances between cluster centers and their assigned trajectories. While they show NP-hardness for determining whether an input curve can be covered with respect to the Fréchet distance, they also present a $O(\log^2 n)$-approximation algorithm for clustering $\kappa$-packed curves (for some constant $\kappa$) under the discrete Fr\'echet distance, where $n$ denotes the total complexity of the input. The overall running time of their algorithm is roughly quadratic in $n$, cubic in $\kappa$ and depends logarithmically on the spread of the vertex coordinates.

In our paper, we focus on the clustering formulation previously studied by  Akitaya, Chambers, Brüning, and Driemel~\cite{akitaya2021subtrajectory}. They present a pseudo-polynomial algorithm that computes a bi-criterial approximation in the sense of Definition~\ref{def:approx} with expected running time in $\tilde{O}(k(\frac{\lambda}{\Delta})^2+\frac{\lambda}{\Delta} n)$, where $\lambda$ denotes the total arclength of the input trajectory. The algorithm finds an $(\alpha,\beta)$-approximate solution with $\alpha\in O(1)$ and $\beta\in O(\ell^2 \log(k\ell))$. It should be noted that in this problem formulation some complexity constraint on the eligible cluster centers is needed to prevent the entire input curve being a cluster center in a trivial clustering.